\documentclass[11pt]{amsart}
\usepackage[english]{babel}
\usepackage{physics}
\usepackage{braket}
\usepackage{dsfont}
\usepackage[table,xcdraw]{xcolor}
\usepackage{multirow}
\usepackage{graphicx}
\usepackage{tikz}
\usepackage{amssymb}
\usepackage{mathtools}
\usetikzlibrary{matrix}

\usepackage[margin=1in]{geometry}
\usepackage[pagebackref, colorlinks = true, linkcolor = blue, urlcolor  = blue, citecolor = red]{hyperref}
\usepackage[utf8x]{inputenc}

\renewcommand{\epsilon}{\varepsilon}

\def\rn{\rho_{\mathcal{N}}}
\def\sn{\sigma_{\mathcal{N}}}
\def\rt{\rho_{\mathcal{T}}}
\def\st{\sigma_{\mathcal{T}}}

\newtheorem{thm}{Theorem}[section]
\newtheorem*{thm*}{Theorem}
\newtheorem{lem}[thm]{Lemma}
\newtheorem{cor}[thm]{Corollary}
\newtheorem{ex}[thm]{Example}
\newtheorem{defi}[thm]{Definition}
\newtheorem{prop}[thm]{Proposition}
\newtheorem{remark}[thm]{Remark}

\mathtoolsset{showonlyrefs}

\begin{document}

\author{Andreas Bluhm}
\email{andreas.bluhm@ma.tum.de}
\address{Zentrum Mathematik, Technische Universit\"at M\"unchen, Boltzmannstrasse 3, 85748 Garching, Germany}

\author{{\'A}ngela Capel }
\email{angela.capel@icmat.es}
\address{Instituto de ciencias matem{\'a}ticas (CSIC-UAM-UC3M-UCM), C/ Nicol{\'a}s Cabrera 13--15, Campus de Cantoblanco, 28049, Madrid, Spain}
\title[A strengthened data processing inequality for the BS-entropy]{A strengthened data processing inequality for the Belavkin-Staszewski relative entropy}

\begin{abstract}
In this work, we provide a strengthening of the data processing inequality for the relative entropy introduced by Belavkin and Staszewski (BS-entropy). 
This extends previous results by Carlen and Vershynina for the relative entropy and other standard $f$-divergences. To this end, we provide two new equivalent conditions for the equality case of the data processing inequality for the BS-entropy. Subsequently, we extend our result to a larger class of maximal $f$-divergences. Here, we first focus on quantum channels which are conditional expectations onto subalgebras and use the Stinespring dilation to lift our results to arbitrary quantum channels. 
\end{abstract}

\date{\today}

\maketitle

\tableofcontents

\section{Introduction}

Quantum $f$-divergences are important in quantum information theory, because they can be used to quantify the similarity of quantum states. Therefore, they fulfill fundamental properties such as data processing, since the distinguishabilty of quantum states cannot increase under the application of a quantum channel. The most important such $f$-divergence is the relative entropy which is defined as
\begin{equation*}
D(\sigma\|\rho) := \tr[\sigma(\log \sigma - \log \rho)]
\end{equation*}
for positive definite quantum states $\rho$, $\sigma$. The relative entropy is one example of the so-called \emph{standard $f$-divergences} \cite[Section 3.2]{Hiai2017}, which are defined as
\begin{equation*}
S_f(\sigma \| \rho) = \tr[\rho^{1/2} f(L_{\sigma} R_{\rho^{-1}})(\rho^{1/2})]
\end{equation*}
for an operator convex function $f:(0, \infty) \to \mathbb R$. Here, $L_{A}$ and $R_{A}$ denote the left and right multiplication by the matrix $A$, respectively. The relative entropy arises by letting $f(x) = x \log x$. This is, however, not the only way to generalize the classical $f$-divergences introduced in \cite{Ali1966, Csiszar1967}. The \emph{maximal $f$-divergences} are defined as
\begin{equation*}
\hat S_f(\sigma \| \rho) = \tr[\rho f(\rho^{-1/2} \sigma \rho^{-1/2})]
\end{equation*}
for an operator convex  function $f:(0, \infty) \to \mathbb R$ and were defined in \cite{Petz1998}.
They were recently studied in \cite{Matsumoto2018} where also the name was introduced (see also \cite[Section 3.3]{Hiai2017} and references therein). For $f(x) = x \log x$, we obtain after a short computation the relative entropy introduced by Belavkin and Staszewski in \cite{Belavkin1982}, which we will call  \textit{BS-entropy} for short:
\begin{equation*} 
\hat S_{\mathrm{BS}}(\sigma \| \rho) := -\tr[\sigma \log(\sigma^{-1/2}\rho \sigma^{-1/2})].
\end{equation*}

It is known that both the standard and maximal $f$-divergences satisfy data processing, i.e.\ they decrease under the application of quantum channels (completely positive trace-preserving maps). The study of conditions for equality in the data processing inequality for the relative entropy, i.e.\ for which $\rho$, $\sigma$ we have $D(\sigma\|\rho) = D(\Phi(\sigma)\|\Phi(\rho))$ for some fixed quantum channel $\Phi$, has led to the discovery of quantum Markov states \cite{Hayden2004}. In particular, the relative entropy is preserved if and only if $\rho$ and $\sigma$ can be recovered by the Petz recovery map $\mathcal{R}_\Phi^\sigma(X) = \sigma^{1/2} \Phi^{\ast}(\Phi(\sigma)^{-1/2} X \Phi(\sigma)^{-1/2})\sigma^{1/2}$, i.e.\ $\sigma = \mathcal{R}_\Phi^\sigma(\Phi(\sigma))$ and $\rho = \mathcal{R}_\Phi^\sigma(\Phi(\rho))$ \cite{Petz2003}. This is true for all standard $f$-divergences for which $f$ is ``complicated enough''. We refer the reader to \cite[Theorem 3.18]{Hiai2017} for a list of equivalent conditions. For $\Phi=\mathcal E$ and $\mathcal E$ the trace-preserving conditional expectation onto a unital matrix subalgebra $\mathcal N$ of $\mathcal B(\mathcal H)$, \cite{Carlen2017a} shows that the equality condition is stable in the sense that
\begin{equation} \label{eq:vc_main_result}
D(\sigma\|\rho) - D(\sigma_{\mathcal N}\|\rho_{\mathcal N}) \geq  \left( \frac{\pi}{8} \right)^4 \norm{L_{\rho} R_{\sigma^{-1}}}^{-2}_\infty \norm{\mathcal{R}_\Phi^\sigma(\rho_{\mathcal N}) - \rho}_1^4.
\end{equation}
Here, we have written $\sigma_{\mathcal N} := \mathcal E (\sigma)$ and $\rho_{\mathcal N}: = \mathcal E(\rho)$. This can also be interpreted as a strengthening of the data processing inequality.
Subsequent work has generalized the above result to more general standard $f$-divergences \cite{Carlen2018} and Holevo's just-as-good fidelity \cite{Wilde2018}. 

The difference of relative entropies that appears on the left hand side of Equation \eqref{eq:vc_main_result} has been studied intensively in the context of quantum information and quantum thermodynamics \cite{Faist2018, Faist2018a}. Moreover, for $\mathcal{E}$ a partial trace, it has been characterized as a conditional relative entropy in \cite{Capel2018}. Equation \eqref{eq:vc_main_result} is the first strengthening of the data processing inequality for the relative entropy in terms of the ``distance'' between a state and its recovery by the Petz map, although there have been many other results with a similar spirit in the last years, see e.g\ \cite{Fawzi2015, Berta2017, Sutter2017b,  Sutter2017, Junge2018}.

Our work gives analogous results to the ones in \cite{Carlen2017a, Carlen2018} for maximal $f$-divergences. For these, preservation of the maximal $f$-divergence, i.e.\ $\hat{S}_f(\Phi(\sigma) \| \Phi(\rho)) = \hat{S}_f(\sigma \| \rho)$, is not equivalent to $\sigma$, $\rho$ being recoverable in the sense of Petz, although the latter implies the former. Equivalent conditions to the preservation of a maximal $f$-divergence for the case in which $\Phi  $ is a  completely positive trace-preserving map are given in \cite[Theorem 3.34]{Hiai2017}. In our work, we prove two other equivalent conditions, which we  use to prove a strengthened data processing inequality for some maximal $f$-divergences and in particular for the BS-entropy. 

This manuscript is structured as follows: In Section \ref{sec:main_results}, we list the main results of the present article. Important results on standard and maximal $f$-divergences are reviewed in Section \ref{sec:preliminaries}. In Section \ref{sec:equality_condition}, we provide two new conditions which are equivalent to the preservation of the BS-entropy under a quantum channel. We use this result in Section \ref{sec:data_processing} to prove our strengthened data processing inequality for the BS-entropy under a conditional expectation, which we subsequently generalize to other maximal $f$-divergences in Section \ref{sec:maximal_f_divergences}. Finally, in Section \ref{sec:quantum-channels}, we extend this result to general quantum channels.

\section{Main results} \label{sec:main_results}
In this section, we state the main results of this work. All quantum systems appearing will be finite dimensional. Let  $\sigma$, $\rho$ be two positive definite quantum states on a matrix algebra $\mathcal M$. We use the abbreviations $\Gamma := \sigma^{-1/2}\rho \sigma^{-1/2}$ and $\Gamma_{\mathcal T} := \sigma_{\mathcal T}^{-1/2}\rho_{\mathcal T} \sigma_{\mathcal T}^{-1/2}$, where $\mathcal N$ is another matrix algebra, $\mathcal T: \mathcal{M} \rightarrow \mathcal{N}$ is a  completely positive trace-preserving map and $\rho_{\mathcal T} := \mathcal T(\rho)$, $\sigma_{\mathcal T} := \mathcal T(\sigma)$. Our first result consists of two conditions which are equivalent to the preservation of the BS-entropy under $\mathcal T$. It follows from Theorem \ref{thm:condition-equality} together with Proposition  \ref{prop:strangecond_implies_BSrecovery} and Theorem \ref{thm:bound-data-processing-BS-channels}.

\begin{thm} \label{thm:new_equality_condition}Let $\mathcal M$ and $\mathcal{N}$ be two matrix algebras and let
$\sigma > 0$, $\rho > 0$ be two quantum states on  $\mathcal M$. Let
 $\mathcal T: \mathcal M \to \mathcal N$ be a completely positive trace-preserving map and let $V$ be the isometry associated to the Stinespring dilation (Theorem \ref{thm:stinespring}) of $\mathcal{T}$.  Then, the following are equivalent:\begin{enumerate}
		\item $\hat S_\mathrm{BS}(\sigma \| \rho) = \hat S_\mathrm{BS}(\sigma_{\mathcal T} \| \rho_{\mathcal T})$
		\item $\label{eq:new_max_equality}
		\sigma^{-1} \rho =  \mathcal{T}^* \left( \sigma_{\mathcal T}^{-1} \rho_{\mathcal T} \right) $
		\item $V \, \sigma^{1/2} \, V^* \left( \st^{-1/2} \, \Gamma_\mathcal{T}^{1/2} \, \st^{1/2} \otimes I  \right) = V \, \Gamma^{1/2} \, \sigma^{1/2} \,  V^*$.

	\end{enumerate}
\end{thm}
The above theorem is motivated by the treatment in \cite{Petz2003} for the relative entropy and proceeds along similar lines. This result allows us to prove a strengthened data processing inequality for the BS-entropy, building on the work in  \cite{Carlen2017a} for conditional expectations and subsequently lifting it to general quantum channels using Stinespring's dilation theorem:

\begin{thm}
Let $\mathcal M$ and $\mathcal{N}$ be two matrix algebras and let $\mathcal T: \mathcal M \to \mathcal N$ be a completely positive trace-preserving map. Let $\sigma$, $\rho$ be two quantum states on $\mathcal{M}$ such that they have equal support. Then,
\begin{equation} \label{eq:BSdata-processing-bound-section2}
\hat S_\mathrm{BS}(\sigma \| \rho) - \hat S_\mathrm{BS}(\sigma_{\mathcal T} \| \rho_{\mathcal T}) \geq  \left( \frac{\pi}{8} \right)^4 \norm{\Gamma}_\infty^{-4} \norm{\st^{-1}}_\infty^{-2}     \norm{ \sigma \, \mathcal{T}^* \left( \st^{-1} \rt \right) -  \rho}_2^4.
\end{equation}
\end{thm}

Theorem \ref{thm:new_equality_condition} shows that the right hand side of Equation \eqref{eq:BSdata-processing-bound-section2} plays the same role as the trace distance between $\rho$ and 
the state obtained from the recovery map in Equation \eqref{eq:vc_main_result}. The result for conditional expectations appears as Corollary \ref{cor:nicer_lower_bound} in the main text and follows from the sharper lower bound in Theorem \ref{thm:bound-data-processing-BS}. These results are subsequently lifted to general quantum channels in Theorem \ref{thm:bound-data-processing-BS-channels}.

 In the rest of the work, we extend the result from the BS-entropy to more general maximal $f$-divergences. This is similar to the work undertaken in \cite{Carlen2018}. We consider operator convex functions $f: (0, \infty) \to \mathbb R$ whose transpose $\tilde f(x) := x f(1/x)$ is operator monotone decreasing. Moreover, we assume that the measure $\mu_{ - \tilde f}$ of $- \tilde f$ is absolutely continuous with respect to Lebesgue measure and that there are $C>0$, $\alpha \geq 0$ such that for every $T \geq 1$, the Radon-Nikod{\'y}m derivative is lower bounded by
\begin{equation*}
\frac{\mathrm{d} \mu_{- \tilde f}(t)}{\mathrm{d}t} \geq \left(C T^{2\alpha}\right)^{-1}
\end{equation*}
almost everywhere (with respect to Lebesgue measure) for all $t \in [1/T, T]$. Furthermore, we assume that our states  $\sigma > 0$ , $\rho > 0$ are not too far from fulfilling the data processing inequality with respect to $\mathcal E$, i.e.\
\begin{equation} \label{eq:not_too_far_from_dpi_main_results}
\left(\frac{(2 \alpha +1)\sqrt{C}}{4} \frac{( \hat{S}_f (\sigma\| \rho) -\hat{S}_f (\sn\| \rn))^{1/2}}{1+\norm{\Gamma}_\infty}\right)^{\frac{1}{1+\alpha}} \leq 1.
\end{equation}

\begin{thm}
Let $\mathcal M$ and $\mathcal{N}$ be two matrix algebras and let $\mathcal T: \mathcal M \to \mathcal N$ be a completely positive trace-preserving map. Let  $\sigma$, $\rho$ be two quantum states on $\mathcal{M}$ such that they have equal support and let $f:(0,\infty) \to \mathbb R$ be an operator convex function with transpose $\tilde f$. We assume that $\tilde f$ is operator monotone decreasing and such that the measure $\mu_{- \tilde f}$ that appears in Theorem \ref{thm:operator-convex-Bhatia} is absolutely continuous with respect to  Lebesgue measure. Moreover, we assume that  for every $T \geq 1$, there exist constants $\alpha \geq 0$, $C > 0$ satisfying $\mathrm{d}\mu_{ - \tilde f} (t) /\mathrm{d}t \geq (C T^{2 \alpha})^{-1}  $ for all $t \in \left[ 1/T , T \right]$ and such that Equation \eqref{eq:not_too_far_from_dpi_main_results} holds. Then, there is a constant $L_{\alpha} > 0$ such that 
\begin{equation*} 
 \hat{S}_f (\sigma\| \rho) -\hat{S}_f (\st\| \rt)   \geq \frac{L_\alpha}{C} \left( 1 +  \norm{\Gamma}_\infty \right)^{-(4 \alpha + 2)}  \norm{\Gamma}_\infty^{-(2 \alpha + 2)} \norm{\st^{-1}}_\infty^{-(2 \alpha + 2)} \norm{ \rho - \sigma \, \mathcal{T}^* \left(  \st^{-1} \rt \right)  }_2^{4(\alpha +1)}.
\end{equation*}
\end{thm}
For conditional expectations, the above result appears as Corollary \ref{cor:nicer_lower_bound_max} in the main text and follows from the sharper lower bound in Theorem \ref{thm:bound-data-processing-max}. The extension to general quantum channels appears as Theorem \ref{thm:bound-data-processing-max-channels}.

\section{Preliminaries} \label{sec:preliminaries}

\subsection{Notation}

 Throughout the paper, we will denote by $\mathcal{H}$ a finite-dimensional Hilbert space, by $\mathcal{B}(\mathcal{H})$ the algebra of bounded linear operators on it (whose elements we will write using capital latin letters). We will further use greek letters  for the density matrices, or states, whose set we write as 
\begin{equation*}
\mathcal{D}(\mathcal{H}):= \{\rho \in \mathcal B(\mathcal H):  \rho \geq 0, \tr[\rho]=1 \}. 
\end{equation*}

A matrix algebra is a unital subalgebra of $\mathcal B(\mathcal H)$ which is closed under taking adjoints. We will write $\rho > 0$  for matrices which are positive definite, denote by $\rho^{-1}$  the inverse of $\rho$, and we will assume that all the density matrices that appear in all sections but Section \ref{sec:quantum-channels} have full rank. By $\rho^0$, we will denote the support of $\rho$. Moreover, we will denote by $I$ the identity matrix, which we sometimes omit for readability. 

Let $\mathcal M$, $\mathcal N$ be two matrix algebras. We call a linear map $\mathcal{T}: \mathcal M \to \mathcal N $  positive if it maps positive semidefinite matrices to positive semidefinite matrices. Moreover,  $\mathcal T$ is said to be $n$-positive if $\mathcal{T} \otimes \operatorname{Id}_n : \mathcal M \otimes \mathcal{M}_n \to \mathcal N \otimes \mathcal{M}_n $ is positive, where $ \mathcal{M}_n $ is the space of complex $n \times n$ matrices, and completely positive if $\mathcal T$ is $n$-positive for every $n \in \mathbb{N}$. We further say that $\mathcal{T}$ is trace preserving if $\tr[\mathcal{T}(A)]= \tr[A]$ for all $A \in \mathcal M$. Finally, a trace-preserving completely positive map is called a quantum channel (a general reference for quantum channels is \cite{Watrous2018}).
Given a quantum channel $\mathcal{T}: \mathcal{M} \rightarrow \mathcal{N}$, we write $\st:= \mathcal{T} (\sigma)$ and $\rt:= \mathcal{T} (\rho)$ for all density matrices $\sigma$, $\rho \in \mathcal{M}$.

Now let $\mathcal{N}$ be  a unital matrix subalgebra of $\mathcal M$. We will often consider the unique  trace-preserving conditional expectation onto this subalgebra, and denote it by $\mathcal{E} : \mathcal{M } \to \mathcal{N}$. Given two matrices $\sigma$, $\rho \in \mathcal{M}$, we write $\sn:= \mathcal{E} (\sigma)$ and $\rn:= \mathcal{E} (\rho)$ . In a slight abuse of notation, we will identify $L_A$, the left multiplication operator by $A$ on $\mathcal M$, with $A$ for $A \in \mathcal M$.

In this paper, we will make use of Schatten $p$-norms, which are defined for every $p \geq 1$ as 
\begin{equation*}
\norm{A}_p := \tr[|A|^p]^{1/p} \; \; \text{ for }A \in \mathcal{B}(\mathcal{H}).
\end{equation*}
In the limit, the $\infty$-norm can be seen to correspond to the operator norm. We will always consider $\mathcal{B}(\mathcal{H})$ equipped with the Hilbert-Schmidt inner product, i.e.,
\begin{equation*}
\left\langle A, B \right\rangle_{\operatorname{HS}} := \tr[A^* B] \; \; \text{ for }A, B \in \mathcal{B}(\mathcal{H}),
\end{equation*}
where $A^*$ represents the adjoint of $A$. We will on most occasions omit the subscript of the inner product. Finally, we will often write $\Gamma := \sigma^{-1/2} \rho \sigma^{-1/2}$, $\Gamma_{\mathcal N} := \sigma_{\mathcal N}^{-1/2} \rho_{\mathcal N} \sigma_{\mathcal N}^{-1/2}$ for conditional expectations and $\Gamma_{\mathcal T} := \sigma_{\mathcal T}^{-1/2} \rho_{\mathcal T} \sigma_{\mathcal T}^{-1/2}$ for quantum channels in general.

\subsection{Operator convex functions and conditional expectations}

Now we will introduce some results concerning operator convex functions and conditional expectations that we use in this manuscript. We refer the reader to \cite[Section V]{Bhatia1997} for further information on the topic of operator convex functions. Before introducing operator convex functions, let us first consider operator monotone functions.
\begin{defi}[Operator monotone] Let $\mathcal{I} \subseteq \mathbb R$ be an interval and $f: \mathcal{I} \to \mathbb R$. If for all finite-dimensional Hilbert spaces $\mathcal H$
\begin{equation*}
f(A) \leq f(B)
\end{equation*} 
for all Hermitian $A$, $B \in \mathcal B(\mathcal H)$ with spectrum contained in $\mathcal{I}$ and such that $A \leq B$, then $f$ is \emph{operator monotone}. We call $f$ \emph{operator monotone decreasing} if $-f$ is operator monotone.
\end{defi}
These functions possess a canonical form:
\begin{thm}[{\cite[Equation V.49]{Bhatia1997}}]\label{thm:operator-convex-Bhatia}
A function $f$ on $(0, \infty)$ is operator monotone if and only if it has a representation of the form 
\begin{equation*}
f(\lambda) = \alpha + \beta \lambda + \int_0^\infty \left( \frac{t}{1+t^2} - \frac{1}{\lambda + t} \right) \mathrm{d}\mu_f(t),
\end{equation*}
where $\alpha \in \mathbb R$, $\beta \geq 0$ and $\mu_f$ is a positive measure on $(0,\infty)$ such that 
\begin{equation*}
\int_0^\infty \frac{1}{1+t^2} \mathrm{d}\mu_f(t) < \infty. 
\end{equation*}
\end{thm}
 Operator monotone functions are intimately connected to operator convex functions.
\begin{defi}[Operator convex] Let $\mathcal{I} \subseteq \mathbb R$ be an interval and $f: \mathcal{I} \to \mathbb R$. If 
\begin{equation*}
f( \lambda A + (1-\lambda) B) \leq \lambda f(A) + (1-\lambda) f(B)
\end{equation*} 
for all Hermitian $A$, $B \in \mathcal B(\mathcal H)$ with spectrum contained in $\mathcal{I}$, all $\lambda \in [0,1]$, and for all finite-dimensional Hilbert spaces $\mathcal H$, then $f$ is \emph{operator convex}. We call $f$ \emph{operator concave} if $-f$ is operator convex.
\end{defi}

One connection of this kind is given by the following theorem:
\begin{thm}[{\cite[Theorem 4.43]{Hiai2014}}] \label{thm:operator_montone_iff_concave}
Let $f$ be a continuous function mapping $(0,\infty)$ onto itself. Then, $f$ is operator monotone if and only if it is operator concave.
\end{thm}
For further connections between operator monotone functions and operator convex functions, we refer to \cite[Section V]{Bhatia1997}. Another way to find new operator convex functions is to consider their transpose. 
\begin{prop}[{\cite[Proposition A.1]{Hiai2017}}] \label{prop:transpose} Let $f:(0,\infty) \to \mathbb R$ and let $\tilde f(x) = x f(1/x)$ for all $x \in (0,\infty)$. Then, $f$ is operator convex if and only if $\tilde f$ is operator convex. $\tilde f$ is called the \emph{transpose} of $f$.
\end{prop}

In the next theorem, we collect several equivalent characterizations of operator convexity. The statements come from {\cite[Theorem 2.1]{Hansen2003}} and \cite[Exercise V.2.2]{Bhatia1997}. 
\begin{thm}[Jensen's operator inequality] \label{thm:jensen}
For a continuous function $f$ defined on an interval $\mathcal{I}$, the following conditions are equivalent: 
\begin{enumerate}
\item $f$ is operator convex on $\mathcal I$.
\item For each natural number $n$, we have the inequality
\begin{equation*}
f\left(\sum_{i = 1}^n A_i^\ast X_i A_i\right) \leq \sum_{i = 1}^n A_i^\ast f(X_i)A_i
\end{equation*}
for every $n$-tuple $(X_1, \ldots, X_n)$ of bounded, self-adjoint operators on an arbitrary Hilbert space $\mathcal H$ with spectra contained in $\mathcal{I}$ and every $n$-tuple $(A_1, \ldots, A_n)$ of operators on $\mathcal H$ with $\sum_{k = 1}^n A_k^\ast A_k = \mathds 1$.
\item $f(V^\ast X V) \leq V^\ast f(X) V$ for every Hermitian operator (on a Hilbert space $\mathcal H$) with spectrum in $\mathcal{I}$ and every isometry $V$ from any Hilbert space into $\mathcal H$.
\end{enumerate}
\end{thm}

\begin{remark}
Let $\mathcal{I} \subseteq \mathbb R$ be an interval, $f: \mathcal{I} \to \mathbb R$ be a continuous operator convex function, $\mathcal M$ be a matrix algebra, and $\Phi: \mathcal M \to \mathcal M$ a unital completely positive map. Then, Jensen's operator inequality in particular implies that $f(\Phi(X)) \leq \Phi(f(X))$ for any Hermitian $X\in \mathcal M$ with spectrum contained in $\mathcal{I}$. This follows from the fact that any completely positive map possesses a Kraus decomposition. 
\end{remark}

We now turn to conditional expectations.

\begin{prop}[{\cite[Proposition 1.12]{Ohya1993}}]
Let $\mathcal M$ be a matrix algebra with unital matrix subalgebra $\mathcal N$. Then, there exists a unique linear mapping $\mathcal E: \mathcal M \to \mathcal N$ such that
\begin{enumerate}
\item $\mathcal E$ is a positive map,
\item $\mathcal E(B) = B$ for all $B \in \mathcal N$,
\item $\mathcal E(AB) = \mathcal E(A)B$ for all $A \in \mathcal M$ and all $B \in \mathcal N$,
\item $\mathcal E$ is trace preserving.
\end{enumerate}
A map fulfilling $(1)$-$(3)$ is called a \emph{conditional expectation}.
\end{prop}
It can be shown that conditional expectations are completely positive \cite[Proposition 5.2.2]{Benatti2009}. Moreover, they are selfadjoint with respect to the Hilbert-Schmidt inner product. The following proposition shows that unital positive maps preserve positive definiteness. In particular, this holds for conditional expectations.

\begin{prop}
Let $\mathcal M$, $\mathcal N$ be two matrix algebras. Moreover, let $\mathcal T: \mathcal M \to \mathcal N$ be a unital positive map.  Then, for $\rho \in \mathcal M$, $\rho > 0$, it holds that $\mathcal T(\rho) > 0$.
\end{prop}
\begin{proof}
Since $\mathcal T$ is a positive map, it holds that $\mathcal T(\rho) \geq 0$. Assume that $\mathcal T(\rho)$ is not positive definite. Then, there is a non-zero $\psi \in \mathcal H$ such that $\tr[\psi \psi^\ast \mathcal T (\rho)] = 0$. However, 
\begin{equation*}
\tr[\psi \psi^\ast \mathcal T(\rho)] = \tr[\mathcal T^\ast(\psi \psi^\ast) \rho].
\end{equation*}
Since $\mathcal T^\ast$ is also a positive map, $\mathcal T^\ast(\psi \psi^\ast) \geq 0$. Furthermore, $\mathcal T^\ast(\psi \psi^\ast) \neq 0$, as $ \mathcal T^\ast$ is trace preserving since $\mathcal T$ is unital. Hence, $\tr[\mathcal T^\ast(\psi \psi^\ast) \rho] > 0$, which is a contradiction.
\end{proof}
 
A standard result for completely positive maps which we will use in Section \ref{sec:quantum-channels} is the existence of a Stinespring dilation (see e.g.\ \cite[Theorem 2.22]{Watrous2018}). It allows us to write a general quantum channel as the action of a conditional expectation $\tr_{ \mathcal V}[\cdot]/s \otimes I$ and an isometry $V$.
\begin{thm}[Stinespring's dilation theorem] \label{thm:stinespring}
 Let $\mathcal M \subseteq \mathcal B(\mathcal H)$, $\mathcal N\subseteq \mathcal B(\mathcal K)$ be two matrix algebras with Hilbert spaces $\mathcal H$, $\mathcal K$, and let $\mathcal T: \mathcal M \to \mathcal N$ be a quantum channel. Then, there exist a Hilbert space $\mathcal V$ and an isometry $V: \mathcal H \hookrightarrow \mathcal K \otimes \mathcal V$ such that 
\begin{equation*}
\mathcal{T} (\omega) = \tr_{\mathcal V} \left[ V \omega V^* \right]
\end{equation*}
for all states $\omega$ on $\mathcal M$. Here, $\tr_{\mathcal V}$ is the partial trace over the second system $\mathcal V$. 
\end{thm}

\subsection{Standard $f$-divergences}

In this subsection, we recall some definitions and basic properties concerning standard $f$-divergences. The main reference for them, as well as for maximal $f$-divergences is \cite{Hiai2017}. The latter are introduced in the next subsection. We focus on states with full rank and refer the reader to \cite{Hiai2017} for a more general treatment.

\begin{defi}[{\cite[Definition 3.1]{Hiai2017}}]
Let $f:(0,\infty) \to \mathbb R$ be an operator convex function and $\sigma > 0$, $\rho > 0$ be two unnormalized states on a matrix algebra $\mathcal M$. Then, 
\begin{equation*}
S_f(\sigma\|\rho) = \tr[\rho^{1/2}f(L_\sigma R_{\rho^{-1}})\rho^{1/2}]
\end{equation*} 
is the \emph{standard $f$-divergence}. This definition can be extended to general states $\sigma, \rho$ as
\begin{equation*}
S_f(\sigma\|\rho) := \underset{\varepsilon \searrow 0}{\operatorname{lim}} \, S_f (\sigma + \varepsilon I \| \rho + \varepsilon I).
\end{equation*}
\end{defi}
We obtain the same standard $f$-divergence if we exchange $\rho$ and $\sigma$ and consider the transpose of $f$ instead.
\begin{prop}[{\cite[Proposition 3.7]{Hiai2017}}]\label{prop:transpose_standard}
Let $f:(0,\infty) \to \mathbb R$ be an operator convex function with transpose $\tilde f$ and $\sigma > 0$, $\rho > 0$ be two states on a matrix algebra $\mathcal M$. Then, $S_f(\sigma\|\rho) = S_{\tilde f}(\rho\|\sigma)$.
\end{prop}

As we can see below, the main examples of standard $f$-divergences are directly connected to the well-known Umegaki relative entropy and standard Rényi divergences.

\begin{ex}[{\cite[Example 3.5]{Hiai2017}}]
Let $f(x) = s(\alpha)x^\alpha$ for some $\alpha \in (0, \infty)$, where $s(\alpha) := -1$ for $0 < \alpha < 1$ and $s(\alpha) := 1$ for $\alpha \geq 1$. Then,
\begin{equation*}
S_f(\sigma \| \rho) = s(\alpha) \tr[\sigma^{\alpha} \rho^{1-\alpha}].
\end{equation*}

These quantities can be used to define the standard R\'enyi divergences. 
\end{ex}

\begin{ex}[{\cite[Example 3.5]{Hiai2017}}]
Let $f(x) = x \log x$. Then,
\begin{equation*} 
S_f(\sigma\|\rho) =  \tr[\sigma (\log \sigma - \log \rho)] 
\end{equation*}
defines the standard (Umegaki) relative entropy, usually denoted by $D(\sigma \| \rho)$. 
\end{ex}

Standard $f$-divergences extend the usual quantum relative entropy in more than one sense, since they share many of the properties that characterize the former, such as continuity (with respect to the first variable) or joint convexity. Indeed, one of the main features of this family of quantities is the data processing inequality.

\begin{prop}[Data processing, {\cite[Proposition 3.12]{Hiai2017}}] \label{prop:standard_dpi}
Let $\Phi: \mathcal M \to \mathcal B $ be a trace-preserving map between matrix algebras $\mathcal M$ and $\mathcal B$ such that its dual map is a 2-positive trace-preserving map. Then, for every two states $\sigma > 0$, $\rho > 0$  on $\mathcal M$ and every operator convex function $f: (0, \infty) \to \mathbb R$,
\begin{equation} \label{eq:data_processing}
S_f(\Phi(\sigma)\|\Phi(\rho)) \leq S_f(\sigma \| \rho). 
\end{equation}
\end{prop}

The above proposition in particular holds for quantum channels. Let us now define the following  map \cite[Equation (3.19)]{Hiai2017} for $\Phi$ as in Proposition \ref{prop:standard_dpi}:
\begin{equation*} 
\mathcal{R}_\Phi^\rho  (X):= \rho^{1/2} \Phi^\ast\left(\Phi(\rho)^{-1/2} ( X ) \Phi(\rho)^{-1/2}\right)\rho^{1/2} \qquad \forall X \in \mathcal B. 
\end{equation*}
This is the Petz recovery map for $\Phi$ with respect to $\rho$. In the following, we will assume that $\Phi$ preserves invertibility, as this will be the case in the situations addressed in this paper.

A natural question is to ask for conditions for when the data processing inequality \eqref{eq:data_processing} holds with equality. Theorem 3.18 of \cite{Hiai2017} gives a list of equivalent conditions, from which we only state some:

\begin{thm}[{\cite[Theorem 3.18]{Hiai2017}}] \label{thm:all-about-std-f}
Let $\sigma > 0$, $\rho > 0$  be two states on a matrix algebra $\mathcal M$ and let $\Phi: \mathcal M \to \mathcal B$ be a $2$-positive trace-preserving linear map, where $\mathcal B$ is again a matrix algebra. Then, the following are equivalent:
\begin{enumerate}
\item There exists a trace-preserving map $\Psi: \mathcal B \to \mathcal M$ such that $\Psi(\Phi(\rho)) = \rho$ and $\Psi(\Phi(\sigma)) = \sigma$.
\item $S_f(\Phi(\sigma) \| \Phi(\rho)) = S_f(\sigma \| \rho)$ for some operator convex function on $(0, \infty)$ such that $f(0^+) < \infty$ and 
\begin{equation*}
|\mathrm{supp}~\mu_f| \geq |\mathrm{spec}(L_\sigma R_{\rho^{-1}}) \cup \mathrm{spec}( L_{\Phi(\sigma)} R_{\Phi(\rho)^{-1}})|, 
\end{equation*}
with $\mu_f$ the measure appearing in \cite[Theorem 8.1]{Hiai2011}.
\item $S_f(\Phi(\sigma)\|\Phi(\rho)) = S_f(\sigma \| \rho)$  for all operator convex $f$ on $[0,\infty)$.
\item $\mathcal{R}_\Phi^\rho(\Phi(\sigma)) = \sigma$.
\end{enumerate}

\end{thm}

In particular, point $(1)$ of Theorem \ref{thm:all-about-std-f} is symmetric in $\sigma$ and $\rho$ such that we obtain the following result, which was previously proven by Petz \cite{Petz2003}.

\begin{cor}\label{thm:Petzrecovery}

Let $\sigma > 0$, $\rho > 0$  be two states on a matrix algebra $\mathcal M$ and let $\Phi: \mathcal M \to \mathcal B$ be a $2$-positive trace-preserving linear map, where $\mathcal B$ is a matrix algebra. Then,

\begin{equation*}
D(\sigma \| \rho) = D( \Phi (\sigma)  \| \Phi( \rho)) \Leftrightarrow  \sigma = \mathcal{R}_\Phi^\rho  \circ \Phi (\sigma). 
\end{equation*}

Moreover, 
\begin{equation*}
\sigma = \mathcal{R}_\Phi^\rho  \circ \Phi (\sigma) \Leftrightarrow \rho = \mathcal{R}_\Phi^\sigma \circ \Phi (\rho). 
\end{equation*}

\end{cor}
 
 \subsection{Maximal $f$-divergences}
 
In this subsection, we introduce maximal $f$-divergences and present some of their most basic properties. We also compare them to the aforementioned standard $f$-divergences. Again, we focus on states with full rank and refer the reader to \cite{Hiai2017} for the general case.
 
\begin{defi}[{\cite[Definition 3.21]{Hiai2017}}] Let $f:(0,  \infty) \to \mathbb R$ be an operator convex function and  $\sigma > 0$, $\rho > 0$ be two states unnormalized on a matrix algebra $\mathcal M$. Then,
\begin{equation*}
\hat S_f(\sigma \| \rho ) = \tr[\rho^{1/2}f(\rho^{-1/2} \sigma \rho^{-1/2})\rho^{1/2}]
\end{equation*}
is the \emph{maximal $f$-divergence}. This  definition can be extended to not necessarily full-rank  states $\sigma, \rho$ as
\begin{equation*}
\hat S_f(\sigma\|\rho) := \underset{\varepsilon \searrow 0}{\operatorname{lim}} \, \hat S_f (\sigma + \varepsilon I \| \rho + \varepsilon I).
\end{equation*}
\end{defi}
Again, the maximal $f$-divergences are identical if we exchange the states and replace $f$ by its transpose.
\begin{prop}[{\cite[Proposition 3.27]{Hiai2017}}] \label{prop:max_transpose}
Let $f:(0,  \infty) \to \mathbb R$ be an operator convex function with transpose $\tilde f$ and  $\sigma > 0$, $\rho > 0$ be two states on a matrix algebra $\mathcal M$. Then, $\hat S_f(\sigma \| \rho )$ = $\hat S_{\tilde f}(\rho \| \sigma )$.
\end{prop}

The main example of a maximal $f$-divergence is the so-called BS-entropy, introduced by Belavkin and Staszewski in \cite{Belavkin1982}.

\begin{ex}[{\cite[Example 4.4]{Hiai2017}}]
Let $f(x) = x \log x$. Then,
\begin{equation*}
\hat S_f(  \sigma \| \rho) = \tr[\rho^{1/2} \sigma \rho^{-1/2}\log(\rho^{-1/2} \sigma \rho^{-1/2})] = \tr[\sigma \log(\sigma^{1/2}\rho^{-1}\sigma^{1/2})]
\end{equation*}
is the \emph{Belavkin-Staszewski relative entropy (BS-entropy)}. 
\end{ex}

Throughout this manuscript, we will use $\hat{S}_{BS}(\cdot \| \cdot)$ to denote the BS-entropy. However, it is common to find in the literature the notation $D_{BS}(\cdot \| \cdot)$ for this quantity.

Remarkably, this family of $f$-divergences also satisfies a data processing inequality, which makes them interesting quantities for information processing.

\begin{prop}[Data processing, {\cite[Corollary 3.31]{Hiai2017}}]\label{prop:data-processing-Maxfdiv}
Let $\sigma > 0$, $\rho > 0$  be two states on a matrix algebra $\mathcal M$  and $\Phi: \mathcal M \to \mathcal B$ be a trace-preserving positive linear map, where $\mathcal B$ is a matrix algebra. Then,
\begin{equation*}
\hat S_f(\Phi(\sigma) \| \Phi(\rho)) \leq \hat S_f(\sigma \| \rho). 
\end{equation*}
\end{prop}

As in the case of standard $f$-divergences, a natural question that arises is to characterize the states for which equality is fulfilled in the previous inequality. Some equivalent conditions for equality are collected in the following result, extracted from the larger list that appears in Theorem 3.34 of \cite{Hiai2017}.

\begin{thm}[{\cite[Theorem 3.34]{Hiai2017}}]\label{thm:all-about-max-f}
Let  $\sigma > 0$, $\rho > 0$ be two states on a matrix algebra $\mathcal M$  and $\Phi: \mathcal M \to \mathcal B$ be a trace-preserving positive linear map, where $\mathcal B$ is a matrix algebra. Then the following are equivalent: 
\begin{enumerate}
\item $\hat S_f(\Phi(\sigma)\| \Phi(\rho)) = \hat S_f(\sigma \| \rho)$ for some non-linear operator convex function $f$ on $[0, \infty)$.
\item $\hat S_f(\Phi(\sigma) \| \Phi(\rho)) = \hat S_f(\sigma \| \rho)$ for all operator convex functions $f$ on $[0, \infty)$.
\item $\tr[\Phi(\sigma)^2\Phi(\rho)^{-1}]=\tr[\sigma^2 \rho^{-1}]$. \label{item:renyi2}
\end{enumerate}
\end{thm}

\begin{remark} 
The function in point \eqref{item:renyi2} of the above theorem is $S_f(\sigma \| \rho) =\hat S_f(\sigma \| \rho)$ for $f(x) = x^2$. Indeed, it is true that if $f$ is a polynomial of degree at most 2, the maximal and the standard $f$-divergences coincide. 
\end{remark}

Another natural question that arises is whether the conditions listed above are equivalent to those of equality in the data processing inequality for standard $f$-divergences that appeared in Theorem \ref{thm:all-about-std-f}. We will later see that this is not the case in general. The following result shows how standard and maximal $f$-divergences are related for the same operator convex function $f$.

\begin{prop}[{\cite[Proposition 4.1]{Hiai2017}}] For every  two states $\sigma > 0$, $\rho > 0$  on a matrix algebra $\mathcal M$ and every operator convex function $f:(0,\infty) \to \mathbb R$,
\begin{equation} \label{eq:maximal_larger}
S_f(\sigma \| \rho) \leq \hat S_f(\sigma \| \rho). 
\end{equation}
\end{prop}

\begin{remark} \label{rmk:strict_ineq}
When $\sigma$ and $\rho$ commute, given an operator convex function $f$ the maximal $f$-divergence coincides with the standard $f$-divergence, and both of them coincide with the classical ones introduced in \cite{Ali1966, Csiszar1967}. In fact, the inequality \eqref{eq:maximal_larger} is strict for states which do not commute, provided $f$ is ``complicated enough" (\cite[Theorem 4.3]{Hiai2017}). For qubits, this is the case for any function $f$ which is not a polynomial (\cite[Proposition 4.7]{Hiai2017}).
\end{remark}

\begin{remark}
Recoverability easily implies $\hat S_f(\Phi(\sigma) \| \Phi(\rho)) = \hat S_f(\sigma \| \rho)$ .The fact that the left hand side is smaller than or equal to the right hand side follows from the data processing inequality. For the other inequality, we can consider $f(x) = x^2$. Then, $\hat S_f(\sigma \| \rho) = \tr[\sigma^2 \rho^{-1}]$. By assumption, 
\begin{equation*}
\sigma = \rho^{1/2} \Phi^\ast(\Phi(\rho)^{-1/2} \Phi(\sigma) \Phi(\rho)^{-1/2}) \rho^{1/2} 
\end{equation*} and 
\begin{align*}
\tr[\sigma^2 \rho^{-1}] &= \tr[\rho (\Phi^\ast(\Phi(\rho)^{-1/2} \Phi(\sigma) \Phi(\rho)^{-1/2}))^2]\\
& \leq \tr[\Phi(\rho) (\Phi(\rho)^{-1/2} \Phi(\sigma) \Phi(\rho)^{-1/2})^2]\\
&= \tr[\Phi(\sigma)^2 \Phi(\rho)^{-1}]
\end{align*}
 The second line is from Jensen's operator inequality (Theorem \ref{thm:jensen}). 
\end{remark}

\begin{remark}
In general, preservation of maximal $f$-divergences does not imply recoverability by means of the Petz recovery map. However, for unital qubit channels, it does  \cite[Proposition 4.10]{Hiai2017}. This does not contradict Remark \ref{rmk:strict_ineq}, since $\Phi$ can still preserve both maximal and standard $f$-divergences, even if their value is not the same.
\end{remark}

\section{A condition for equality} \label{sec:equality_condition}
 Theorem 3.34 of \cite{Hiai2017} lists several equivalent conditions for the preservation of maximal $f$-divergences under a quantum channel. We will prove two other  equivalent conditions, inspired by \cite{Petz2003}. We need the following technical proposition in the proof of the main result of this section.

\begin{prop}\label{prop:propU*UandGamma}
Let $\mathcal M$ be two matrix algebras. We consider two quantum states  $\sigma > 0$ and $\rho > 0$ on  $\mathcal{M}$  and a completely positive trace-preserving map $\mathcal{T}:\mathcal{M}\rightarrow \mathcal{N}$ such that $\sigma_{\mathcal T}$, $\rho_{\mathcal T} > 0$. Let $U: \mathcal N \to \mathcal M$ be given by $U(X) = \sigma^{1/2} \, \mathcal{T}^* \left( \st^{-1/2} X \right)$ for all $X \in \mathcal 
N$.
Then, $U^*(Y) = \st^{-1/2} \, \mathcal{T}(\sigma^{1/2} Y)$ for every $Y \in \mathcal{M}$ and
\begin{equation*}
U^* \Gamma U \leq \Gamma_\mathcal{T},
\end{equation*}
Moreover, 
$U^\ast U  \leq \operatorname{Id}$. If $\mathcal N$ is a unital subalgebra of $\mathcal M$ and $\mathcal T = \mathcal E$, where $\mathcal E$ is the conditional expectation onto $\mathcal N$, we can extend $U$ to an operator on $\mathcal M$
and it holds that $U^\ast U = \mathcal E$.

\end{prop}

\begin{proof}
The form of $U^\ast$ follows from direct computation. Let $X \in \mathcal N$. Then,
\begin{align*}
\langle X, U^* \Gamma U(X) \rangle &= \langle U(X), \Gamma U(X) \rangle \\
&=\langle \sigma^{1/2} \, \mathcal{T}^* \left( \st^{-1/2} X \right), \sigma^{-1/2} \rho \, \mathcal{T}^* \left( \st^{-1/2} X \right) \rangle \\
&= \tr[ \rho \, \mathcal{T}^* \left( \st^{-1/2} X  \right) \mathcal{T}^* \left( X^* \st^{-1/2}  \right) ] \\
& \leq \tr[ \rho \, \mathcal{T}^* \left( \st^{-1/2} X X^* \st^{-1/2}  \right) ] \\
& = \langle X, \Gamma_{\mathcal{T}}   X \rangle.
\end{align*}
The fourth line follows by the Schwarz inequality. Hence, $U^\ast \Gamma U \leq \Gamma_{\mathcal N}$.
A similar calculation yields 
\begin{align*}
\langle X, U^\ast U(X) \rangle &= \langle U(X), U(X) \rangle \\
&= \langle \sigma^{1/2} \, \mathcal{T}^* \left( \st^{-1/2} X \right), \sigma^{1/2} \, \mathcal{T}^* \left( \st^{-1/2} X \right) \rangle \\
&= \tr[ \sigma \, \mathcal{T}^* \left( \st^{-1/2} X  \right) \mathcal{T}^* \left( X^* \st^{-1/2}  \right)] \\
& \leq \tr[ \sigma \, \mathcal{T}^* \left( \st^{-1/2} X  X^* \st^{-1/2}  \right)]  \\
&= \langle X, X \rangle . 
\end{align*} 
This implies $U^\ast U \leq \operatorname{Id}$. In the case where $\mathcal T$ is a conditional expectation, we can write $U(X) = \sigma^{1/2} \, \sigma_{\mathcal N}^{-1/2} \mathcal E(X) $ for all $X \in \mathcal M$. The equation $U^\ast U = \mathcal E$ then follows from a similar calculation to the one above and the fact that $\mathcal E$ is trace preserving.
\end{proof}

 Now we can state and prove the  new equivalent condition for equality between BS-entropies under the application of a quantum channel.
 
\begin{thm}\label{thm:condition-equality} Let $\mathcal M$, $\mathcal N$ be two matrix algebras and $\mathcal{T} : \mathcal M \rightarrow \mathcal N$ be a completely positive trace-preserving map. Let $\sigma > 0$, $\rho> 0$ be two quantum states on $\mathcal M$ such that $\mathcal{T}(\sigma) >0$, $\mathcal{T}(\rho) >0$. Then
\begin{equation}\label{eq:BSequal} 
\hat S_\mathrm{BS}(\sigma \| \rho) = \hat S_\mathrm{BS}(\st || \rt) 
\end{equation}
if and only if  
\begin{equation}\label{eq:condBSequal}
\mathcal{T}^* \left( \st^{-1}  \rt \right)=  \sigma^{-1} \rho.
\end{equation} 
\end{thm}

\begin{proof}
The proof follows the proof of \cite[Theorem 3.1]{Petz2003}. Let $U: \mathcal N \to \mathcal M$ be defined as $U(X) = \sigma^{1/2} \, \mathcal{T}^* \left( \st^{-1/2} X \right) $ for all $X \in \mathcal N$. Using the integral representation of the operator monotone function $\log(x)$,
\begin{equation*}
\log{x} = \int_{0}^\infty \left(\frac{1}{1+t} - \frac{1}{t+x}\right) \mathrm{d}t,
\end{equation*}
we infer below that Equation \eqref{eq:BSequal} is equivalent to 
\begin{equation}\label{eq:BSequal-inf}
\left\langle \sigma_{\mathcal T}^{1/2}, U^\ast \left( (\Gamma + t)^{-1} - (t+1)^{-1}I \right) U \sigma_{\mathcal T}^{1/2} \rangle = \langle \sigma_{\mathcal T}^{1/2},  \left( (\Gamma_{\mathcal T} + t)^{-1} - (t+1)^{-1}I \right) \sigma_{\mathcal T}^{1/2} \right\rangle .
\end{equation}
Indeed, we know that  $\Gamma_{\mathcal T} \geq U^\ast \Gamma U$  and $U^\ast U \leq \mathrm{Id}$ (see Proposition \ref{prop:propU*UandGamma}).  Let $f_t(x) = (t+x)^{-1} -t^{-1}$ for fixed $t \geq 0$. Since $x \mapsto x^{-1}$ is operator monotone decreasing and operator convex on $(0, \infty)$, the same property holds for $f_t(x)$ on $[0,\infty)$ for $t > 0$. Hence,
\begin{equation*}
(U^\ast\Gamma U + t)^{-1} - t^{-1}I  \geq (\Gamma_{\mathcal T} + t)^{-1} - t^{-1}I .
\end{equation*}
Moreover, $f_t(x)\leq 0$ for every $x \geq 0$. Using  \cite[Theorem V.2.3]{Bhatia1997} and the fact that $U$ is a contraction, it holds that
\begin{equation*}
U^\ast \left( (\Gamma + t)^{-1} - t^{-1}I \right) U  \geq  (U^\ast\Gamma U + t)^{-1} - t^{-1}I ,
\end{equation*}
and thus, 
\begin{equation} \label{eq:integrand_inequality}
U^\ast \left( (\Gamma + t)^{-1} - t^{-1}I \right) U \geq (\Gamma_{\mathcal T} + t)^{-1} - t^{-1}I .
\end{equation}
Hence, since $U (\st^{1/2}) = \sigma^{1/2}$,
\begin{align*}
\hat S_\mathrm{BS}(\sigma \| \rho) - \hat S_\mathrm{BS}(\sigma_{\mathcal T} \| \rho_{\mathcal T}) & = \int_{0}^\infty   \left\langle \sigma^{1/2},  \left( (\Gamma + t)^{-1} - (t+1)^{-1}I \right)  \sigma^{1/2} \right\rangle  \mathrm{d}t \\
 & \phantom{as} -  \int_{0}^\infty   \left\langle \sigma_{\mathcal T}^{1/2},  \left( (\Gamma_\mathcal{T} + t)^{-1} - (t+1)^{-1}I \right)  \sigma_{\mathcal T}^{1/2} \right\rangle \mathrm{d}t \\
& = \int_{0}^\infty   \left\langle \sigma_{\mathcal T}^{1/2}, \left( U^\ast (\Gamma + t)^{-1} U - (\Gamma_\mathcal{T} + t)^{-1} \right)  \sigma_{\mathcal T}^{1/2} \right\rangle   \mathrm{d}t \\
& \geq  0 ,
\end{align*}
where the last inequality follows from Equation \eqref{eq:integrand_inequality}. Moreover, since for every $t > 0$ the infinitesimal term at time $t$ inside the integral is always non-negative, the difference of BS-entropies vanishes if  and only if every infinitesimal term does. Therefore, Equation \eqref{eq:BSequal} is equivalent to Equation \eqref{eq:BSequal-inf}, and they both imply 
\begin{equation*}
 U^\ast (\Gamma + t)^{-1} \sigma^{1/2} = (\Gamma_{\mathcal T} + t)^{-1} \sigma_{\mathcal T}^{1/2}
\end{equation*}
for all $t > 0$. Differentiating with respect to $t$ gives 
\begin{equation*}
 U^\ast (\Gamma + t)^{-2} \sigma^{1/2} = (\Gamma_{\mathcal T} + t)^{-2} \sigma_{\mathcal T}^{1/2}.
\end{equation*}
It follows that
\begin{align*}
\left\|  U^\ast (\Gamma + t)^{-1} \sigma^{1/2} \right\|_2^2 &=   \left\langle  \sigma_{\mathcal T}^{1/2}, (\Gamma_{\mathcal T} + t)^{-2} \sigma^{1/2}_{\mathcal T}\right\rangle \\
&=    \left\langle  \sigma_{\mathcal T}^{1/2},  U^\ast (\Gamma + t)^{-2} \sigma^{1/2} \right\rangle \\
&= \left\|(\Gamma + t)^{-1} \sigma^{1/2} \right\|_2^2.
\end{align*}
We have shown $\langle A, U U^\ast A \rangle = \langle A, A \rangle$ for some $A \in \mathcal M$ and we know $U U^\ast \leq \mathrm{Id}$ since $\norm{U}_\infty = \norm{U^\ast}_\infty$, thus we infer $U U^\ast A= A$. Therefore, we have arrived at
\begin{equation*}
U (\Gamma_{\mathcal T} + t)^{-1} \sigma_{\mathcal T}^{1/2} = U U^\ast (\Gamma + t)^{-1} \sigma^{1/2} =  (\Gamma + t)^{-1} \sigma^{1/2}
\end{equation*}
Differentiating again with respect to $t$, it follows that 
\begin{equation*}
U (\Gamma_{\mathcal T}  + t)^{-n} \sigma_{\mathcal T}^{1/2} =  (\Gamma + t)^{-n} \sigma^{1/2}
\end{equation*}
for all $n \in \mathbb N$ and hence also 
\begin{equation*}
U f(\Gamma_{\mathcal T}) \sigma_{\mathcal T}^{1/2}=  f(\Gamma) \sigma^{1/2}
\end{equation*}
for all continuous functions $f$ by the Stone-Weierstrass theorem. For $f(x) = x$, we obtain
\begin{equation*}
\sigma^{1/2} \mathcal{T}^* \left( \sigma_{\mathcal T}^{-1} \rho_{\mathcal T} \right)= \sigma^{-1/2} \rho.
\end{equation*}
This proves the first implication. The reverse implication follows  from
\begin{equation*}
\tr[\rt^2\, \st^{-1}]= \tr[\rho \, \mathcal T^\ast(\rt \, \st^{-1})] = \tr[\rho^2 \, \sigma^{-1}] 
\end{equation*}
and the fact that
\begin{equation*}
 \tr[\rt^2 \, \st^{-1} ] =  \tr[\rho^2 \, \sigma^{-1}]  \Leftrightarrow  \tr[\st^2 \, \rt^{-1} ] =  \tr[\sigma^2\, \rho^{-1}]
 \end{equation*} 
by Theorem \ref{thm:all-about-max-f} for $f(x)=x^{1/2}$, $\tilde{f}= f(x)$.
\end{proof}

\begin{remark}\label{rem:BS-recovery-condition}
Note that Equation \eqref{eq:condBSequal} can be rephrased as a recovery condition for $\rho$ from $\sigma$ under the application of a quantum channel:
 \begin{equation*}
\rho = \sigma \,  \mathcal{T}^* \left( \st^{-1} \rt \right),
\end{equation*}
as well as exchanging the roles of $\rho$ and $\sigma$.
\end{remark}

\begin{remark}
In the particular case in which the map is a trace-preserving conditional expectation $\mathcal{E}$ onto a unital matrix subalgebra $\mathcal{ N}$ of $\mathcal{M}$, Theorem \ref{thm:condition-equality} can be written as follows:
\begin{equation*}
\hat{S}_{\operatorname{BS}}(\sigma || \rho ) = \hat{S}_{\operatorname{BS}}(\sn|| \rn ) 
\end{equation*}
if and only if
\begin{equation*}
\sn^{-1} \rn = \sigma^{-1} \rho.
\end{equation*}
Here, we have assumed $\sigma >0$, $\rho >0$. In this case, the recovery condition for $\rho$ from $\sigma$ under the application of a conditional expectation is stated as follows:
\begin{equation*}
\rho = \sigma \sn^{-1} \rn.
\end{equation*}
\end{remark}

We can further see that, for quantum channels, the condition appearing in Equation \eqref{eq:BSequal} is implied by another one involving $\Gamma$ and $\Gamma_\mathcal{T}$ which will appear in the main result of Section  \ref{sec:quantum-channels}.

\begin{prop}\label{prop:strangecond_implies_BSrecovery}
Let $\mathcal{M}$ be a matrix algebra and let $\sigma >0$, $\rho >0$ be two states on it. Let $\mathcal N$ be another matrix algebra and let $\mathcal{T} : \mathcal{M} \rightarrow \mathcal{N}$ be a quantum channel. Let $V$ be the isometry associated to a Stinespring dilation (Theorem \ref{thm:stinespring}) of $\mathcal{T}$. 
If the following expression holds
\begin{equation}\label{eq:strangecondition_channels_stinespring}
V \, \sigma^{1/2} \, V^* \left( \st^{-1/2} \, \Gamma_\mathcal{T}^{1/2} \, \st^{1/2} \otimes I  \right) = V \, \Gamma^{1/2} \, \sigma^{1/2} \,  V^*,
\end{equation}
then
\begin{equation*}
\sigma^{-1} \rho = \mathcal{T}^* \left( \st^{-1} \rt  \right).
\end{equation*}

\end{prop}

\begin{proof}

Using Equation \eqref{eq:strangecondition_channels_stinespring}, and abbreviating $\Theta := \st^{-1/2} \, \Gamma^{1/2} _\mathcal{T} \, \st^{1/2} \otimes I $, we can see that
\begin{align*}
V \, \Gamma \, \sigma^{1/2} \, V^* & = V \, \Gamma^{1/2} \, V^* \, V \, \Gamma^{1/2} \, \sigma^{1/2} \, V^* \\
&= V \, \Gamma^{1/2} \, V^* \, V \, \sigma^{1/2}  \, V^* \, \Theta \\
&= V \, \Gamma^{1/2} \, \sigma^{1/2} \, V^* \, \Theta \\
&= V \, \sigma^{1/2} \, V^* \, \Theta^2 .
\end{align*}
Now, note that 
\begin{equation*}
\Theta^2 = \st^{-1} \rt \otimes I.
\end{equation*}
Hence, multiplying the expression above by $V^*(\cdot) V$ and using $\mathcal T^\ast(X) = V^\ast (X \otimes I) V$ for all $X  \in \mathcal N$, we get
\begin{align*}
\Gamma \, \sigma^{1/2} & = \sigma^{1/2} V^* \left( \st^{-1} \rt \otimes I  \right) V \\
& = \sigma^{1/2} \mathcal{T}^* \left( \st^{-1} \rt  \right) , 
\end{align*}
which is equivalent to
\begin{equation*}
\sigma^{-1} \rho = \mathcal{T}^* \left( \st^{-1} \rt  \right).
\end{equation*}
\end{proof}

\begin{remark}
The converse implication is also true, although we cannot prove it directly here. However, it can be obtained as a consequence of Theorem \ref{thm:bound-data-processing-BS-channels}. Note also that multiplying directly Equation \eqref{eq:strangecondition_channels_stinespring}  by $V^\ast(\cdot)V$, we get the following expression:
\begin{equation*}
\sigma^{1/2} \, V^* \left( \st^{-1/2} \, \Gamma_\mathcal{T} \, \st^{1/2} \otimes I  \right) V = \Gamma^{1/2} \, \sigma^{1/2},
\end{equation*}
which can be rewritten as
\begin{equation}\label{eq:strangecondition_channels}
\sigma^{1/2} \, \mathcal{T}^* \left(  \st^{-1/2} \, \Gamma_{\mathcal{T}}^{1/2} \, \st^{1/2}  \right) = \Gamma^{1/2} \, \sigma^{1/2}.
\end{equation}  
For conditional expectations, this condition can be actually seen to be equivalent to Equation \eqref{eq:BSequal}.
\end{remark}

\begin{prop}\label{prop:BS-equality}
Let $\mathcal{M}$ be a matrix algebra, $\mathcal{N}$ be a unital matrix subalgebra, and $\mathcal{E}: \mathcal{ M} \rightarrow \mathcal{N}$ be the unique trace-preserving conditional expectation onto $\mathcal N$. Let $\sigma >0$, $\rho >0$ and define $\sn:= \mathcal{E}(\sigma)$, $\rn:= \mathcal{E}(\rho)$. Then, 
\begin{equation}\label{eq:BSrecoverability}
\rho = \sigma \sn^{-1} \rn
\end{equation}
is equivalent to
\begin{equation}\label{eq:Strange-expression}
\sigma^{1/2} \sn^{-1/2} \Gamma_\mathcal{N}^{1/2} \sn^{1/2} = \Gamma^{1/2} \sigma^{1/2}. 
\end{equation}
\end{prop}

\begin{proof}

Recalling the explicit expressions for $\Gamma$ and $\Gamma_\mathcal{N}$, Equation \eqref{eq:BSrecoverability} can be seen to be equivalent to
\begin{equation*}
\sigma^{1/2} \sn^{-1/2} \Gamma_\mathcal{N} = \Gamma \sigma^{1/2} \sn^{-1/2},
\end{equation*} 
and iterating $n$ times, we get
\begin{equation*}
\sigma^{1/2} \sn^{-1/2} \Gamma^n_\mathcal{N} = \Gamma^n \sigma^{1/2} \sn^{-1/2}.
\end{equation*} 
By the Weierstrass theorem, this implies 
\begin{equation*}
\sigma^{1/2} \sn^{-1/2} f( \Gamma_\mathcal{N} ) = f( \Gamma) \sigma^{1/2} \sn^{-1/2},
\end{equation*} 
for every continuous function $f$, and, in particular, for $f(x)=x^{1/2}$, we have
\begin{equation}\label{eq:strange-equiv-condition}
\sigma^{1/2} \sn^{-1/2}  \Gamma_\mathcal{N}^{1/2} = \Gamma^{1/2} \sigma^{1/2} \sn^{-1/2}. 
\end{equation} 
This concludes \eqref{eq:BSrecoverability} $\implies$ \eqref{eq:Strange-expression}. The converse implication follows from Equation \eqref{eq:strange-equiv-condition}, iterating it twice.
\end{proof}

Equation \eqref{eq:strange-expression} will appear in the main result of Section \ref{sec:data_processing}. As a direct consequence of Theorem \ref{thm:condition-equality}  for conditional expectations and Proposition \ref{prop:BS-equality}, we have the following result.

\begin{cor}\label{cor:BS-equality}
Under the conditions of the proposition above, the following facts are equivalent:
\begin{enumerate}
\item $\hat S_{BS} ( \sigma \| \rho ) =  \hat S_{BS} (\sn \| \rn) $.
\item $\rho = \sigma \sn^{-1} \rn$.
\item  $\sigma^{1/2} \sn^{-1/2} \Gamma_\mathcal{N}^{1/2} \sn^{1/2} = \Gamma^{1/2} \sigma^{1/2}$. 
\end{enumerate}
\end{cor}

Let us denote the aformentioned asymmetric recovery map, which we will call \textit{BS recovery condition} throughout the rest of the paper, by
\begin{equation*}
\mathcal{B}_\mathcal{T}^\sigma (\cdot) := \sigma \mathcal{T^*} \left( \st^{-1} (\cdot) \right).
\end{equation*}
Note that, although $ \mathcal{B}_\mathcal{T}^\sigma$ is trace-preserving, it is not completely positive in general. Moreover, analogously to Theorem \ref{thm:Petzrecovery}, Theorem \ref{thm:condition-equality} can be restated as
\begin{equation}\label{eq:equiv-BS-recovery-condition}
\hat S_\mathrm{BS}( \sigma \| \rho ) = \hat S_\mathrm{BS}(\sigma_{\mathcal T} \| \rt ) \Leftrightarrow  \rho = \mathcal{B}_\mathcal{T}^\sigma \circ \mathcal{T}  (\rho).
\end{equation}

\begin{remark}
Note that, analogously to the case for the relative entropy, from Remark \ref{rem:BS-recovery-condition} and Equation \eqref{eq:equiv-BS-recovery-condition} we can deduce
\begin{align*}
\hat S_\mathrm{BS}( \sigma \| \rho )  = \hat S_\mathrm{BS}(\sigma_{\mathcal T} \|\rt )   & \Leftrightarrow   \rho  = \mathcal{B}_\mathcal{T}^\sigma \circ \mathcal{T}  (\rho) \\ 
& \Leftrightarrow   \sigma  = \mathcal{B}_\mathcal{T}^\rho \circ \mathcal{T}  (\sigma) \\
& \Leftrightarrow \hat S_\mathrm{BS}(  \rho \| \sigma  )  = \hat S_\mathrm{BS}( \rt  \| \st ).  
\end{align*}
Here, the second equivalence follows from Theorem \ref{thm:all-about-max-f} and the fact that $\tilde f(x) = f(x)$ for $f(x) = x^{1/2}$.
\end{remark}

Now, a natural question is whether $\sigma$ can be recovered in the sense of Petz in the same cases that it can be recovered in the sense of the BS-entropy, and thus, whether the conditions of equality for the relative entropy coincide with those of equality for the BS-entropy. This can be answered negatively in general, although one implication always holds.

Indeed, from \cite[Theorem 2]{Petz2003}, we can see that $D(\sigma \| \rho) = D (\st \| \rt)$ is equivalent to 

\begin{equation*}
\mathcal{T}^* \left( \st^{it} \rt^{-it} \right)=  \sigma^{it} \rho^{-it}  \text{ for every }t \in \mathbb{R},
\end{equation*}
and by analytic continuation, it implies
\begin{equation*}
\mathcal{T}^* \left( \st^{z} \rt^{-z} \right)=  \sigma^{z} \rho^{-z}  \text{ for every }z \in \mathbb{C}.
\end{equation*}
In particular, 
\begin{equation*}
 D( \sigma \| \rho ) = D ( \st \| \rt ) \implies \mathcal{T}^* \left( \st^{-1} \rt \right)=  \sigma^{-1} \rho,
\end{equation*}
obtaining the following well-known result:

\begin{cor}
Let $\sigma$, $\rho > 0$ be states on $\mathcal{M}$ and such that $\st$, $\rt > 0$ for $\mathcal{T}: \mathcal{M} \to \mathcal{N}$ a quantum channel. Then,
\begin{equation*}
 D( \sigma \| \rho ) = D ( \st \| \rt ) \implies  \hat{S}_{BS}( \sigma \| \rho )  = \hat{S}_{BS} (\st \| \rt).
\end{equation*}
Equivalently, 
\begin{equation*}
\sigma = \mathcal{R}_\mathcal{T}^\rho \circ \mathcal{T} (\sigma) \implies \sigma = \mathcal{B}_\mathcal{T}^\rho \circ \mathcal{T} (\sigma). 
\end{equation*}
\end{cor}

The converse implications are false in general. Indeed, \cite[Example 2.2]{Jencova2009} and \cite[Example 4.8]{Hiai2017} constitute examples of states for which there is equality between BS-entropies but one state  cannot be recovered from the other using the Petz recovery map.

\section{Strengthened data processing inequality for the BS-entropy} \label{sec:data_processing}

The well-known data processing inequality for the partial trace, whose extension for standard $f$-divergences is Proposition \ref{prop:standard_dpi}, finds its analogue for maximal $f$-divergences in Proposition \ref{prop:data-processing-Maxfdiv}. In the main result of this section, inspired by \cite{Carlen2017a}, we will prove a strengthened lower bound for the data processing inequality for the BS-entropy when the map considered is a trace-preserving conditional expectation onto a unital matrix subalgebra $\mathcal N$ of $\mathcal M$. We will present an extension  of this result to general quantum channels in Section \ref{sec:quantum-channels}.  Before we start, we introduce some important tools.

\begin{lem}\label{lemma:Lemma2.1CarlenVershynina}
Let $\mathcal M$ be a matrix algebra with unital subalgebra $\mathcal N$. Let $\sigma > 0$, $\rho> 0$ be two quantum states on $\mathcal{M}$ and consider $\mathcal E: \mathcal M \to \mathcal N$  the unique trace-preserving conditional expectation onto this subalgebra. 
Consider $U: \mathcal{M} \to \mathcal{M}$ defined as in Proposition \ref{prop:propU*UandGamma}. Then
\begin{equation*}
 \left\langle	\sn^{1/2}, \left( U^\ast (\Gamma + t)^{-1} U - (\Gamma_\mathcal{N} + t)^{-1}  \right) \sn^{1/2} \right\rangle   \geq t \norm{ \left( U (\Gamma_\mathcal{N} + t)^{-1} - (\Gamma + t)^{-1} U \right) \sn^{1/2} }_2^2 ,
\end{equation*}
for every $t>0$.
\end{lem}

\begin{proof}
By virtue of \cite[Lemma 2.1]{Carlen2017a}, we know that 
\begin{equation*}
 \left\langle	\sn^{1/2},  U^\ast (\Gamma + t)^{-1} U   \sn^{1/2} \right\rangle =  \left\langle	\sn^{1/2} , (\Gamma_\mathcal{N} + t)^{-1}   \sn^{1/2} \right\rangle + \left\langle  w_t , (\Gamma + t) w_t \right\rangle,
\end{equation*}
for 
\begin{equation*}
w_t := U (\Gamma_\mathcal{N} + t)^{-1} \sn^{1/2} - (\Gamma + t)^{-1} U \sn^{1/2}.
\end{equation*}
Hence, taking into account that 
\begin{equation*}
\left\langle  w_t , (\Gamma + t) w_t \right\rangle \geq t \norm{w_t}_2^2,
\end{equation*}
we get
\begin{equation*}
 \left\langle	\sn^{1/2}, \left( U^\ast (\Gamma + t)^{-1} U - (\Gamma_\mathcal{N} + t)^{-1}  \right) \sn^{1/2} \right\rangle   \geq t \norm{ \left( U (\Gamma_\mathcal{N} + t)^{-1} - (\Gamma + t)^{-1} U \right) \sn^{1/2} }_2^2 .
\end{equation*}

\end{proof}
We need another tool before we can prove the main result of this section.
 
\begin{prop}
Consider two quantum states $\rho$, $\sigma > 0$ on $\mathcal{M}$ and their expectations $\rn$ and $\sn$ on  $\mathcal{N} \subset \mathcal{M}$.  Define  $\Gamma = \sigma^{-1/2} \rho \sigma^{-1/2}$ and $\Gamma_\mathcal{N} = \sn^{-1/2} \rn \sn^{-1/2}$. Then,
\begin{equation*}\label{prop:Spectra_gammas}
\| \Gamma_{\mathcal N} \|_{\infty} \leq \| \Gamma \|_{\infty}. 
\end{equation*}
\end{prop}

\begin{proof}
Let us introduce the norm $\| A\|_{\infty, \mathcal A}$ for $\mathcal A$ some unital subalgebra of $\mathcal B(\mathcal H)$ and $A: \mathcal A \to \mathcal B(\mathcal H)$ a linear map. The norm is defined as
\begin{equation*}
\| A\|_{\infty, \mathcal A} := \sup_{B \in \mathcal A} \frac{\| A(B) \|_{2}}{\| B\|_{2}}.
\end{equation*}
We note that $\mathcal N$ and $\mathcal M$ form a Hilbert space with the Hilbert-Schmidt norm and the bounded operators on this Hilbert space form a C$^\ast$-algebra with the above norms (for $\mathcal A = \mathcal M$ and $\mathcal N$, respectively). Furthermore,
\begin{equation*}
\| \Gamma_{\mathcal N} \|_{\infty, \mathcal M} = \| \Gamma_{\mathcal N} \|_{\infty, \mathcal N} = \| \sigma_{\mathcal N}^{-1/2} \rho_{\mathcal N}\sigma_{\mathcal N}^{-1/2} \|_{\infty},
\end{equation*}
since
\begin{equation*}
\| \Gamma_{\mathcal N} \|_{\infty, \mathcal M} \leq \sup_{B \in \mathcal M}   \frac{\| \sigma_{\mathcal N}^{-1/2} \rho_{\mathcal N}\sigma_{\mathcal N}^{-1/2} \|_{\infty} \| B\|_{2}}{\| B\|_{2}}
\end{equation*}
and 
\begin{equation*}
\| \Gamma_{\mathcal N} \|_{\infty, \mathcal N} \geq \frac{\| \sigma_{\mathcal N}^{-1/2} \rho_{\mathcal N}\sigma_{\mathcal N}^{-1/2} P \|_{2}} {\norm{P}_{2}} = \| \sigma_{\mathcal N}^{-1/2} \rho_{\mathcal N}\sigma_{\mathcal N}^{-1/2} \|_{\infty},
\end{equation*}
where $P$ is the projection on the eigenspace of the largest eigenvalue of $\sigma_{\mathcal N}^{-1/2} \rho_{\mathcal N}\sigma_{\mathcal N}^{-1/2}$. As $\mathcal N$ is a von Neumann algebra, it holds that $P \in \mathcal N$ (see \cite[Section 2.4.2]{Bratteli1979}). Proposition \ref{prop:propU*UandGamma} shows that $\Gamma_{\mathcal N} = U^\ast \Gamma U$ on $(\mathcal N, \langle \cdot, \cdot \rangle_{\operatorname{HS}})$. Thus,
\begin{align*}\| \sigma_{\mathcal N}^{-1/2} \rho_{\mathcal N}\sigma_{\mathcal N}^{-1/2} \|_{\infty} & = \| \Gamma_{\mathcal N} \|_{\infty, \mathcal N} \\
& = \| U^\ast\Gamma U \|_{\infty, \mathcal N} \\
& \leq \| U \|_{\infty, \mathcal M}^2 \| \Gamma \|_{\infty, \mathcal M} \\
& \leq \| \Gamma \|_{\infty}.
\end{align*}
The last line follows, since $U^\ast U = \mathcal E$, $\mathcal E \leq \operatorname{Id}$ and therefore $\norm{U(B)}^2_2 \leq \langle B, \mathcal E(B) \rangle \leq \norm{B}_2^2$ for all $B \in \mathcal M$.
\end{proof} 
 
The main result of this section reads as follows.

\begin{thm}\label{thm:bound-data-processing-BS}
Let $\mathcal M$ be a matrix algebra with unital subalgebra $\mathcal N$. Let $\mathcal E: \mathcal M \to \mathcal N$ be the trace-preserving conditional expectation onto this subalgebra. Let $\sigma > 0$, $\rho> 0$ be two quantum states on $\mathcal{M}$. Then
\begin{equation} \label{eq:BSdata-processing-bound}
\hat S_\mathrm{BS}(\sigma \| \rho) - \hat S_\mathrm{BS}(\sigma_{\mathcal N} \| \rho_{\mathcal N}) \geq \left( \frac{\pi}{4} \right)^4 \norm{\Gamma}_\infty^{-2}  \norm{ \sigma^{1/2} \sn^{-1/2} \Gamma_\mathcal{N}^{1/2} \sn^{1/2} - \Gamma^{1/2}  \sigma^{1/2}}_2^4.
\end{equation}
\end{thm}

\begin{proof}
The first part of the proof follows the first part of the one of Theorem \ref{thm:condition-equality}. Consider $U: \mathcal{M} \to \mathcal{M}$ as defined in Proposition \ref{prop:propU*UandGamma}. Then, the following inequality holds as operators on $(\mathcal N, \langle \cdot, \cdot \rangle_{\operatorname{HS}})$
\begin{equation*}
U^\ast \left( (\Gamma + t)^{-1} - (t+1)^{-1}I \right) U \geq (\Gamma_{\mathcal N} + t)^{-1} - (t+1)^{-1}I .
\end{equation*}
Therefore, 
\begin{align*}
\hat S_\mathrm{BS}(\sigma \| \rho) & = \int_{0}^\infty   \left\langle \sigma_{\mathcal N}^{1/2}, U^\ast \left( (\Gamma + t)^{-1} - (t+1)^{-1}I \right) U \sigma_{\mathcal N}^{1/2} \right\rangle   \mathrm{d}t \\
& \geq \int_{0}^\infty   \left\langle \sigma_{\mathcal N}^{1/2},  \left( (\Gamma_\mathcal{N} + t)^{-1} - (t+1)^{-1}I \right)  \sigma_{\mathcal N}^{1/2} \right\rangle \mathrm{d}t \\
&= \hat S_\mathrm{BS}(\sigma_{\mathcal N} \| \rho_{\mathcal N}). 
\end{align*}
Consider the infinitesimal expressions in the previous integrals. Hence,  given $0<T< \infty$, following the proof of \cite[Theorem 1.7]{Carlen2017a} and by virtue of the Cauchy-Schwarz inequality
\begin{align*}
\hat S_\mathrm{BS} (\sigma \| \rho)  - \hat S_\mathrm{BS}(\sigma_{\mathcal N} \| \rho_{\mathcal N}) & \geq \int_0^T \left\langle	\sn^{1/2}, \left( U^\ast (\Gamma + t)^{-1} U - (\Gamma_\mathcal{N} + t)^{-1}  \right) \sn^{1/2} \right\rangle \mathrm{d}t \\
& \geq \int_0^T t \norm{ \left( U (\Gamma_\mathcal{N} + t)^{-1} - (\Gamma + t)^{-1} U \right) \sn^{1/2} }_2^2 \mathrm{d}t \\
& \geq \frac{1}{T} \left( \int_0^T t^{1/2}  \norm{ \left( U (\Gamma_\mathcal{N} + t)^{-1} - (\Gamma + t)^{-1} U \right) \sn^{1/2} }_2 \mathrm{d}t   \right)^2.
\end{align*}
Here, we have used Lemma \ref{lemma:Lemma2.1CarlenVershynina} in the second line. Let us study the expression appearing in the last integral. For that, recall the integral representation of the operator monotone square root function,
\begin{equation*}
x^{1/2} = \frac{1}{\pi} \int_0^\infty t^{1/2} \left( \frac{1}{t} - \frac{1}{t + x}   \right) \mathrm{d}t, 
\end{equation*}
which clearly yields
\begin{equation*}
U \Gamma_\mathcal{N}^{1/2} \sn^{1/2} - \Gamma^{1/2} U \sn^{1/2} = \frac{1}{\pi} \int_0^\infty t^{1/2} \left(   (\Gamma + t)^{-1} U  - U (\Gamma_\mathcal{N} + t)^{-1} \right) \sn^{1/2} \mathrm{d}t. 
\end{equation*}
The left hand side can be simplified as
\begin{equation*}
U \Gamma_\mathcal{N}^{1/2} \sn^{1/2} - \Gamma^{1/2} U \sn^{1/2} = \sigma^{1/2} \sn^{-1/2} \Gamma_\mathcal{N}^{1/2} \sn^{1/2} - \Gamma^{1/2}  \sigma^{1/2} ,
\end{equation*}
and thus
\begin{align*}
\norm{\sigma^{1/2} \sn^{-1/2} \Gamma_\mathcal{N}^{1/2} \sn^{1/2} - \Gamma^{1/2}  \sigma^{1/2}}_2  & = \frac{1}{\pi}  \norm{ \int_0^\infty t^{1/2} \left( U (\Gamma_\mathcal{N} + t)^{-1} - (\Gamma + t)^{-1} U \right)\sn^{1/2} \mathrm{d}t}_2 \\
& \leq \frac{1}{\pi}  \int_0^T t^{1/2} \norm{ \left( U (\Gamma_\mathcal{N} + t)^{-1} - (\Gamma + t)^{-1} U \right) \sn^{1/2} }_2  \mathrm{d}t\\
& \phantom{as}+ \frac{1}{\pi}  \norm{ \int_T^\infty t^{1/2} \left( U (\Gamma_\mathcal{N} + t)^{-1} - (\Gamma + t)^{-1} U \right)\sn^{1/2} \mathrm{d}t}_2
\end{align*}
for any $0< T < \infty$. We present now an upper bound for the last term on the right hand side. As shown in the proof of \cite[Theorem 1.7]{Carlen2017a}, 
\begin{align*}
&  \norm{\int_T^\infty t^{1/2} \left( U (\Gamma_\mathcal{N} + t)^{-1} - (\Gamma + t)^{-1} U \right)\sn^{1/2} \mathrm{d}t }_2 \\
& \phantom{ad} \leq \norm{ \int_T^\infty t^{1/2} \left( U (\Gamma_\mathcal{N} + t)^{-1} - t^{-1} U \right)\sn^{1/2} \mathrm{d}t}_2 +  \norm{\int_T^\infty t^{1/2} \left( U  t^{-1} - (\Gamma + t)^{-1} U \right)\sn^{1/2}\mathrm{d}t }_2
\end{align*}
Moreover, we have 
\begin{equation*}
 \int_T^\infty t^{1/2} \left(t^{-1} I -  (\Gamma_\mathcal{N} + t)^{-1} \right) \mathrm{d}t \leq  \frac{2 \norm{\Gamma_\mathcal{N}}_\infty}{T^{1/2}} I 
\end{equation*}
and
\begin{equation*}
\int_T^\infty t^{1/2} \left(  t^{-1} I - (\Gamma + t)^{-1}\right) \mathrm{d}t  \leq \frac{2 \norm{\Gamma}_\infty}{T^{1/2}} I.
\end{equation*}
Thus, 
\begin{equation*}
\norm{\int_T^\infty t^{1/2} \left( U (\Gamma_\mathcal{N} + t)^{-1} - (\Gamma + t)^{-1} U \right)\sn^{1/2} \mathrm{d}t }_2 \leq \frac{4 \norm{\Gamma}_\infty}{T^{1/2}} , 
\end{equation*}
since $U^\ast U \leq \mathrm{Id}$ by Proposition \ref{prop:propU*UandGamma}, $\norm{\sigma_\mathcal{N}^{1/2}}_2 = 1$, and $\norm{\Gamma_{\mathcal N}}_\infty \leq \norm{\Gamma}_\infty$ by Proposition \ref{prop:Spectra_gammas}. Therefore,
\begin{equation*}
\norm{ \sigma^{1/2} \sn^{-1/2} \Gamma_\mathcal{N}^{1/2} \sn^{1/2} - \Gamma^{1/2}  \sigma^{1/2}}_2  \leq \frac{1}{\pi} T^{1/2}  \left( \hat S_\mathrm{BS} (\sigma\| \rho)  - \hat S_\mathrm{BS}(\sn \| \rn) \right)^{1/2}  + \frac{4 \norm{\Gamma}_\infty}{\pi T^{1/2}} .
\end{equation*}
Optimizing this expression with respect to $T$, we find the optimal bound 
\begin{equation*}\label{eq:strange-expression}
\norm{ \sigma^{1/2} \sn^{-1/2} \Gamma_\mathcal{N}^{1/2} \sn^{1/2} - \Gamma^{1/2}  \sigma^{1/2}}_2  \leq \frac{4 \norm{\Gamma}^{1/2}_\infty}{\pi}  \left( \hat S_\mathrm{BS} (\sigma\| \rho)  - \hat S_\mathrm{BS}(\sn \| \rn) \right)^{1/4} . 
\end{equation*}
Finally, rearranging the terms, we obtain Equation \eqref{eq:BSdata-processing-bound}.
\end{proof}

We have obtained a lower bound for the difference of BS-entropies in terms of one expression that already appeared in the previous section, in Corollary \ref{cor:BS-equality}. Furthermore, we can find another lower bound for it  with an expression that provides a measure of the recoverability of $\rho$ in terms of the relation found in Theorem \ref{thm:condition-equality}.

\begin{lem}\label{lem:BS_data_processing_bound}
Let $\mathcal M$ be a matrix algebra with unital subalgebra $\mathcal N$. Let $\mathcal E: \mathcal M \to \mathcal N$ be the trace-preserving conditional expectation onto this subalgebra. Let $\rho> 0$, $\sigma > 0$ be two quantum states on $\mathcal{M}$. Then,
\begin{equation*}
\norm{ \sigma^{1/2} \sn^{-1/2} \Gamma_\mathcal{N}^{1/2} \sn^{1/2} - \Gamma^{1/2}  \sigma^{1/2}}_2  \geq \frac{1}{2}  \norm{\Gamma}_\infty^{-1/2}   \norm{\sigma^{-1}}_\infty^{-1/2}    \norm{\sigma \sn^{-1} \rn - \rho}_2 .
\end{equation*}
\end{lem}

\begin{proof}

Let us define
\begin{equation*}
A := \sigma^{1/2} \sn^{-1/2} \Gamma_\mathcal{N}^{1/2} \sn^{1/2} - \Gamma^{1/2}  \sigma^{1/2}.
\end{equation*}
It holds that $\norm{\sn^{-1/2}}_\infty \leq \norm{\sigma^{-1}}_\infty^{1/2}$ by Jensen's operator inequality and the Russo-Dye theorem. Using the facts that $\norm{\sigma^{1/2}}_2 = \norm{\sn^{1/2}}_2=1$, on the one side we have
\begin{align*}
& \norm{ \sigma^{1/2}  \sn^{-1/2}  \Gamma_\mathcal{N} - \Gamma \sigma^{1/2} \sn^{-1/2}}_2   = \\
& \phantom{adadadasadad}  = \norm{ \sigma^{1/2}  \sn^{-1/2}  \Gamma_\mathcal{N} -  \Gamma^{1/2} \sigma^{1/2} \sn^{-1/2} \Gamma_\mathcal{N}^{1/2} +  \Gamma^{1/2} \sigma^{1/2} \sn^{-1/2} \Gamma_\mathcal{N}^{1/2} - \Gamma \sigma^{1/2} \sn^{-1/2}}_2 \\
& \phantom{adadadasadad} \leq \norm{ \sigma^{1/2}  \sn^{-1/2}  \Gamma_\mathcal{N} -  \Gamma^{1/2} \sigma^{1/2} \sn^{-1/2} \Gamma_\mathcal{N}^{1/2} }_2 + \norm{ \Gamma^{1/2} \sigma^{1/2} \sn^{-1/2} \Gamma_\mathcal{N}^{1/2} - \Gamma \sigma^{1/2} \sn^{-1/2}}_2   \\
& \phantom{adadadasadad}  \leq \norm{A}_2 \norm{\sigma_{\mathcal N}^{-1/2}}_\infty \left(\norm{\Gamma_\mathcal{N}^{1/2}}_\infty + \norm{\Gamma^{1/2}}_\infty\right) \\
& \phantom{adadadasadad} \leq  2  \norm{A}_2 \norm{\sigma^{-1}}_\infty^{1/2}  \norm{\Gamma}_\infty^{1/2},
\end{align*} 
where we have used Hölder's inequality  and Proposition \ref{prop:Spectra_gammas}. On the other side, we get
\begin{align*}
\norm{ \sigma^{1/2}  \sn^{-1/2}  \Gamma_\mathcal{N} - \Gamma \sigma^{1/2} \sn^{-1/2}}_2  & =  \norm{ \sigma^{1/2}  \sn^{-1}  \rn \sn^{-1/2} - \sigma^{-1/2} \rho  \sn^{-1/2}}_2  \\
 & \geq \norm{\sigma \sn^{-1} \rn - \rho}_2.
\end{align*} 
Therefore, 
\begin{equation*}
\norm{\sigma \sn^{-1} \rn - \rho}_2 \leq 2 \norm{\Gamma}^{1/2}_\infty  \norm{\sigma^{-1}}_\infty^{1/2}  \norm{\sigma^{1/2} \sn^{-1/2} \Gamma_\mathcal{N}^{1/2} \sn^{1/2} - \Gamma^{1/2}  \sigma^{1/2}}_2.
\end{equation*}

\end{proof}

Notice that $\norm{\sigma^{-1}}_\infty$ is nothing but the inverse of the minimum eigenvalue of $\sigma$. Finally, as a consequence of Theorem \ref{thm:bound-data-processing-BS} and Lemma \ref{lem:BS_data_processing_bound}, we get the following corollary.

\begin{cor} \label{cor:nicer_lower_bound}
Let $\mathcal M$ be a matrix algebra with unital subalgebra $\mathcal N$. Let $\mathcal E: \mathcal M \to \mathcal N$ be the trace-preserving conditional expectation onto this subalgebra. Let  $\sigma > 0$, $\rho> 0$ be two quantum states on $\mathcal{M}$. Then,
\begin{equation} \label{eq:BSdata-processing-bound2}
\hat S_\mathrm{BS}(  \sigma \| \rho) - \hat S_\mathrm{BS}(\sn \| \rn)  \geq \left( \frac{\pi}{8} \right)^4 \norm{\Gamma}_\infty^{-4}  \norm{\sigma^{-1}}_\infty^{-2}   \norm{\rho- \sigma \sn^{-1} \rn  }_2^4 .
\end{equation}
\end{cor}

\begin{remark}
This result, in particular, constitutes another proof for the implication
\begin{equation*}
\hat S_\mathrm{BS}(  \sigma \| \rho) = \hat S_\mathrm{BS}(\sn \| \rn)  \implies \rho = \sigma \sn^{-1} \rn,
\end{equation*}
from Theorem \ref{thm:condition-equality}. Indeed, we can deduce from the proof of this corollary the implications $(1) \implies (3) \implies (2) $ in Corollary \ref{cor:BS-equality}.

\end{remark}

\section{On the data processing inequality for maximal $f$-divergences}\label{sec:maximal_f_divergences}

 In this section, we consider a more general setting than in the previous ones and, following the lines of \cite{Carlen2018}, we provide a strengthened data processing inequality for maximal $f$-divergences. 
We consider operator convex functions $f: (0, \infty) \to \mathbb R$ whose transpose $\tilde f$ is operator monotone decreasing. The transpose is operator convex by Proposition \ref{prop:transpose} and it is also monotone decreasing if $f$ maps $(0, \infty) $ to itself by Theorem \ref{thm:operator_montone_iff_concave}. Since the functions we consider here belong to a more general framework, we have to further assume the latter, although the aforementioned theorem shows that it is a reasonable assumption.

Moreover, we demand that the measure $\mu_{- \tilde f}$ of the transpose with negative sign is absolutely continuous with respect to Lebesgue measure and  assume that there are $C>0$, $\alpha \geq 0$ such that for every $T \geq 1$, the Radon-Nikod{\'y}m derivative satisfies
\begin{equation*}
\frac{\mathrm{d}\mu_{ - \tilde f}(t)}{\mathrm{d}t} \geq \left(C T^{2\alpha}\right)^{-1}
\end{equation*}
 almost everywhere (with respect to Lebesgue measure) for all $t \in [1/T, T]$. Moreover, we impose the condition that our states $\sigma$, $\rho > 0$ are such that
\begin{equation} \label{eq:not_too_far_from_dpi}
\left(\frac{(2\alpha + 1)\sqrt{C}}{4} \frac{( \hat{S}_f (\sigma\| \rho) -\hat{S}_f (\sn\| \rn))^{1/2}}{1+\norm{\Gamma}_\infty}\right)^{\frac{1}{1+\alpha}} \leq 1.
\end{equation}
The main result of this section is the following:

\begin{thm}\label{thm:bound-data-processing-max}
Let $\mathcal M$ be a matrix algebra with unital subalgebra $\mathcal N$. Let $\mathcal E: \mathcal M \to \mathcal N$ be the trace-preserving conditional expectation onto this subalgebra. Let  $\sigma > 0$, $\rho> 0$ be two quantum states on $\mathcal{M}$ and let $f:(0,\infty) \to \mathbb R$ be an operator convex function with transpose $\tilde f$. We assume that $\tilde f$ is operator monotone decreasing and such that the measure $\mu_{- \tilde f}$ that appears in Theorem \ref{thm:operator-convex-Bhatia} is absolutely continuous with respect to  Lebesgue measure. Moreover, we assume that  for every $T \geq 1$, there exist constants $\alpha \geq 0$, $C > 0$ satisfying $\mathrm{d}\mu_{ - \tilde f} (t) /\mathrm{d}t \geq (C T^{2 \alpha})^{-1}  $ for all $t \in \left[ 1/T , T \right]$ and such that Equation \eqref{eq:not_too_far_from_dpi} holds. Then, there is a constant $K_{\alpha} > 0$ such that 
\begin{equation} \label{eq:Max-data-processing-bound}
\hat S_f(  \sigma \| \rho) - \hat S_f(\sn \| \rn)   \geq \frac{K_\alpha}{C} \left( 1 +  \norm{\Gamma}_\infty\right)^{-(4 \alpha + 2)}  \norm{ \sigma^{1/2} \sn^{-1/2} \Gamma_\mathcal{N}^{1/2} \sn^{1/2} - \Gamma^{1/2}  \sigma^{1/2}}_2^{4(\alpha + 1)}.
\end{equation}
\end{thm}

\begin{proof}
Recall that, given an operator convex function $f$ with transpose $\tilde f$, 
\begin{equation*}
 \hat{S}_f (\sigma\| \rho)  =  \hat{S}_{\tilde f} (\rho\| \sigma) = \tr[\sigma^{1/2} \tilde f( \Gamma) \sigma^{1/2}]
\end{equation*}
by Proposition \ref{prop:max_transpose}. By assumption, $\tilde f$ is operator monotone decreasing. Thus, by virtue of Theorem \ref{thm:operator-convex-Bhatia}, $- \tilde f$ can be written as
\begin{equation*}
- \tilde f(\lambda) = \alpha + \beta \lambda + \int_0^\infty \left(  \frac{t}{t^2 + 1} - \frac{1}{t + \lambda}  \right) \mathrm{d}\mu_{- \tilde f} (t),
\end{equation*}
with $\alpha \in \mathbb R$, $\beta \geq 0$ and $\mu_{- \tilde f}$ a positive measure on $(0, \infty)$ such that
\begin{equation*}
\int_0^\infty  \frac{1}{t^2 + 1} \mathrm{d}\mu_{- \tilde f} (t) < \infty.
\end{equation*}
Hence, it is clear that 
\begin{align*}
 \hat{S}_f (\sigma\| \rho)  &= \left\langle  \sigma^{1/2}, { \tilde f}(\Gamma) \sigma^{1/2}    \right\rangle \\
& =  \left\langle  \sigma^{1/2},   \left( - \alpha -\beta \Gamma + \int_0^\infty \left( (\Gamma +t)^{-1} - \frac{t}{t^2 + 1} I   \right) \mathrm{d}\mu_{- \tilde f}(t)  \right) \sigma^{1/2}    \right\rangle \\
& = - \alpha - \beta + \int_0^\infty \left\langle  \sigma^{1/2}, \left( ( \Gamma + t)^{-1} - \frac{t}{t^2 + 1} I   \right)  \sigma^{1/2}    \right\rangle \mathrm{d}\mu_{- \tilde f}(t) \\
& \geq - \alpha - \beta + \int_0^\infty \left\langle  \sn^{1/2}, \left( (\Gamma_\mathcal{N} + t)^{-1} - \frac{t}{t^2 + 1} I   \right)  \sn^{1/2}    \right\rangle \mathrm{d}\mu_{- \tilde f}(t) \\
&= \hat{S}_f (\sn\| \rn) ,
\end{align*}
 where the inequality in the fourth line follows from Proposition \ref{prop:propU*UandGamma} and Jensen's operator inequality (point (3) in Theorem \ref{thm:jensen}).
Notice that the difference of maximal $f$-divergences is given by 
\begin{equation*}
\hat{S}_f (\sigma\| \rho) -\hat{S}_f (\sn\| \rn) = \int_0^\infty \left( \left\langle  \sigma^{1/2}, ( \Gamma + t)^{-1} \sigma^{1/2} \right \rangle - \left \langle \sn^{-1/2}, ( \Gamma_\mathcal{N} + t)^{-1}  \sn^{1/2}\right\rangle \right) \mathrm{d}\mu_{- \tilde f}(t), 
\end{equation*} 
and, recalling that $U \sigma_\mathcal{N}^{1/2} = \sigma^{1/2}$, the difference between the infinitesimal terms in the integrals was studied in Lemma \ref{lemma:Lemma2.1CarlenVershynina}, obtaining
\begin{equation*}
 \left\langle	\sn^{1/2}, \left( U^\ast (\Gamma + t)^{-1} U - (\Gamma_\mathcal{N} + t)^{-1}  \right) \sn^{1/2} \right\rangle   \geq t \norm{ \left( U (\Gamma_\mathcal{N} + t)^{-1} - (\Gamma + t)^{-1} U \right) \sn^{1/2} }_2^2 ,
\end{equation*}
Following the proof of Theorem \ref{thm:bound-data-processing-BS}, we infer that 
\begin{equation*}
\norm{\sigma^{1/2} \sn^{-1/2} \Gamma_\mathcal{N}^{1/2} \sn^{1/2} - \Gamma^{1/2}  \sigma^{1/2}}_2   = \frac{1}{\pi}  \norm{ \int_0^\infty t^{1/2} \left( U (\Gamma_\mathcal{N} + t)^{-1} - (\Gamma + t)^{-1} U \right)\sn^{1/2} \mathrm{d}t}_2, 
\end{equation*}
and we can now split the right hand side into three parts (contrary to the proof of Theorem \ref{thm:bound-data-processing-BS}, where we only split it in two):
\begin{align*}
\norm{\sigma^{1/2} \sn^{-1/2} \Gamma_\mathcal{N}^{1/2} \sn^{1/2} - \Gamma^{1/2}  \sigma^{1/2}}_2   &\leq \underbrace{\frac{1}{\pi} \int_0^{1/T} t^{1/2}  \norm{  \left( U (\Gamma_\mathcal{N} + t)^{-1} - (\Gamma + t)^{-1} U \right)\sn^{1/2} }_2 \mathrm{d}t }_{(\ast 1)} \\
& +  \underbrace{\frac{1}{\pi} \int_{1/T}^T t^{1/2}  \norm{  \left( U (\Gamma_\mathcal{N} + t)^{-1} - (\Gamma + t)^{-1} U \right)\sn^{1/2} }_2 \mathrm{d}t }_{(\ast 2)} \\ &+  \underbrace{\frac{1}{\pi}   \norm{  \int_T^\infty t^{1/2}  \left( U (\Gamma_\mathcal{N} + t)^{-1} - (\Gamma + t)^{-1} U \right)\sn^{1/2} \mathrm{d}t}_2 }_{(\ast 3)} .
\end{align*}
Here, we assume that $T \geq 1$. Let us study each one of these terms separately: For the first one, we have
\begin{equation*}
(\ast 1) \leq \frac{2}{\pi}  \int_0^{1/T} t^{-1/2} \mathrm{d}t = \frac{4}{\pi T^{1/2}},
\end{equation*}
since 
\begin{equation*}
\norm{  \left( U (\Gamma_\mathcal{N} + t)^{-1} - (\Gamma + t)^{-1} U \right)\sn^{1/2} }_2 \leq 2 t^{-1}.
\end{equation*}
The last term is bounded using the same reasoning as in the proof of Theorem \ref{thm:bound-data-processing-BS}. Thus, we have
\begin{equation*}
(\ast 3) \leq \frac{4 \norm{\Gamma}_\infty}{\pi T^{1/2}}.
\end{equation*}
The second term, however, introduces something that had not appeared on the main result of the previous section. Indeed, by the Cauchy-Schwarz inequality, we have
\begin{align*}
( \ast 2 ) & \leq \frac{1}{\pi} \left( T - \frac{1}{T} \right)^{1/2}  \left( \int_{1/T}^T t \norm{  \left( U (\Gamma_\mathcal{N} + t)^{-1} - (\Gamma + t)^{-1} U \right)\sn^{1/2} }^2_2 \mathrm{d}t \right)^{1/2} \\
& \leq \frac{T^{1/2 }}{\pi} \sqrt{C}T^\alpha   \left( \int_{1/T}^T t \norm{  \left( U (\Gamma_\mathcal{N} + t)^{-1} - (\Gamma + t)^{-1} U \right)\sn^{1/2} }^2_2 \mathrm{d}\mu_{ - \tilde f}(t)  \right)^{1/2} \\
& \leq \frac{T^{1/2}}{\pi} \sqrt{C}T^\alpha  \left( \hat{S}_f (\sigma\| \rho) -\hat{S}_f (\sn\| \rn) \right)^{1/2} .
\end{align*}
Let us assume that $ \hat{S}_f (\sigma\| \rho) -\hat{S}_f (\sn\| \rn)> 0$. Considering the three bounds together and optimizing over $T$, we  find that the minimum is achieved for
\begin{equation*}
T = \left(\frac{4}{(2 \alpha +1)\sqrt{C}} \frac{1+\norm{\Gamma}_\infty}{( \hat{S}_f (\sigma\| \rho) -\hat{S}_f (\sn\| \rn))^{1/2}}\right)^{\frac{1}{1+\alpha}}.
\end{equation*}
We note that indeed $T \geq 1$ by Equation \eqref{eq:not_too_far_from_dpi}. Inserting the optimal $T$, we obtain

\begin{align*}
&\norm{\sigma^{1/2} \sn^{-1/2} \Gamma_\mathcal{N}^{1/2} \sn^{1/2} - \Gamma^{1/2}  \sigma^{1/2}}_2  \\ &\leq (K_\alpha)^{-\frac{1}{4(\alpha + 1)}} \left( 1 + \norm{\Gamma}_\infty \right)^{\frac{2 \alpha + 1}{2 \alpha + 2}}  C^{\frac{1}{4(\alpha + 1)}}  \left( \hat{S}_f (\sigma\| \rho) -\hat{S}_f (\sn\| \rn) \right)^{\frac{1}{4(\alpha + 1)}} ,
\end{align*}
and rearranging the terms, we get Equation \eqref{eq:Max-data-processing-bound}. Here,
\begin{equation*}
K_\alpha = \left(\frac{2 \alpha + 1}{2 \alpha + 2}\right)^{4(\alpha + 1)} (2 \alpha +1)^{-2} 4^{-(4 \alpha + 2)} \pi^{4 (\alpha+1)}.
\end{equation*}
This bound is also valid for $\hat{S}_f(\sigma\| \rho) -\hat{S}_f (\sn\| \rn) = 0$, since in this case we can make the upper bound arbitrarily small by choosing $T$ arbitrarily large. 
\end{proof}

Lemma \ref{lem:BS_data_processing_bound} can also be used to get another bound for the difference of maximal $f$-divergences in terms of the BS recovery map applied to $\rho$. 

\begin{cor} \label{cor:nicer_lower_bound_max}
Let $\mathcal M$ be a matrix algebra with unital subalgebra $\mathcal N$. Let $\mathcal E: \mathcal M \to \mathcal N$ be the trace-preserving conditional expectation onto this subalgebra. Let  $\sigma > 0$, $\rho> 0$ be two quantum states on $\mathcal{M}$ and let $f:(0,\infty) \to \mathbb R$ be an operator convex function with transpose $\tilde f$. We assume that $\tilde f$ is operator monotone decreasing and such that the measure $\mu_{ - \tilde f}$ that appears in Theorem \ref{thm:operator-convex-Bhatia} is absolutely continuous with respect to  Lebesgue measure. Moreover, we assume that  for every $T \geq 1$, there exist constants $\alpha \geq 0$, $C > 0$ satisfying $\mathrm{d}\mu_{ - \tilde f} (t) /\mathrm{d}t \geq (C T^{2 \alpha})^{-1}  $ for all $t \in \left[ 1/T , T \right]$ and such that Equation \eqref{eq:not_too_far_from_dpi} holds. Then, there is a constant $L_{\alpha} > 0$ such that 
\begin{equation} 
 \hat{S}_f (\sigma\| \rho) -\hat{S}_f (\sn\| \rn)   \geq \frac{L_\alpha}{C} \left( 1 +  \norm{\Gamma}_\infty \right)^{-(4 \alpha + 2)}  \norm{\Gamma}_\infty^{-(2 \alpha + 2)} \norm{\sigma^{-1}}_\infty^{-(2 \alpha + 2)} \norm{ \rho - \sigma  \sn^{-1} \rn  }_2^{4(\alpha +1)}.
\end{equation}
\end{cor}

 As a consequence of Theorem \ref{thm:bound-data-processing-max}, we get the following strengthening of the data processing inequality for maximal $f$-divergences for particular operator convex functions. The first one concerns the BS-entropy. In this case, $f(x) = x \log x$, $\tilde f(x) = - \log x$, $\alpha = 0$ and $C = 1$.

\begin{cor} 
Let $\mathcal M$ be a matrix algebra with unital subalgebra $\mathcal N$. Let $\mathcal E: \mathcal M \to \mathcal N$ be the trace-preserving conditional expectation onto this subalgebra. Let $\sigma > 0$, $\rho> 0$ be two quantum states on $\mathcal{M}$ such that $ \hat{S}_{\operatorname{BS}} (\sigma\| \rho) -\hat{S}_{\operatorname{BS}}  (\sn\| \rn) \leq 4(\norm{\Gamma}_\infty +1)$. Then,
\begin{equation}\label{eq:cor-BS-data-processing-max-f-divergences} 
 \hat{S}_{\operatorname{BS}} (\sigma\| \rho) -\hat{S}_{\operatorname{BS}}  (\sn\| \rn)   \geq \left( \frac{\pi}{8} \right)^4 \left( 1 +  \norm{\Gamma}_\infty \right)^{-2}  \norm{\Gamma}_\infty^{-2} \norm{\sigma^{-1}}_\infty^{-2}  \norm{ \rho - \sigma  \sn^{-1} \rn   }_2^4.
\end{equation}
\end{cor}

Notice that Equation \eqref{eq:cor-BS-data-processing-max-f-divergences} is a bit less tight than Equation \eqref{eq:BSdata-processing-bound2}, although the results are comparable. The next corollary deals with the data processing inequality of maximal $f$-divergences associated to power functions.

\begin{cor}
Let $\mathcal M$ be a matrix algebra with unital subalgebra $\mathcal N$. Let $\mathcal E: \mathcal M \to \mathcal N$ be the trace-preserving conditional expectation onto this subalgebra. Let $\sigma > 0$, $\rho> 0$ be two quantum states on $\mathcal{M}$ and take $f_\beta(x)=  - x^{1-\beta}$, for $0 < \beta <1$. Then, $ C= \frac{\pi}{\sin{\pi \beta}}$, $ \alpha = \beta/2$  and if Equation \eqref{eq:not_too_far_from_dpi} holds, we have: 
\begin{equation}\label{eq:cor-powers-data-processing-max-f-divergences} 
 \hat{S}_{f_{\beta}} (\sigma\| \rho) -\hat{S}_{f_\beta}  (\sn\| \rn)   \geq L_{\beta/2} \frac{\sin{\pi \beta}}{\pi}  \left( 1 +  \norm{\Gamma}_\infty \right)^{-2(\beta +1)}  \norm{\Gamma}_\infty^{-(\beta + 2)} \norm{\sigma^{-1}}_\infty^{-(\beta + 2)}  \norm{ \rho - \sigma  \sn^{-1} \rn  }_2^{2 \beta + 4}.
\end{equation}
\end{cor}

\begin{proof}
This follows straight from the facts $\tilde f(x) = - x^\beta$ and that \cite[Example 3]{Carlen2018}
\begin{equation*}
\mathrm{d}\mu_{f_\beta} (t) = \frac{\sin{\pi \beta}}{\pi} t^{\beta} \mathrm{d}t .
\end{equation*}
An application of Theorem \ref{thm:bound-data-processing-max} and Lemma \ref{lem:BS_data_processing_bound} yields
\begin{equation*}
L_{\beta/2} = \frac{1}{4} \left(\frac{\beta + 1}{\beta + 2}\right)^{2 \beta + 4} (\beta +1)^{-2} 8^{-2(\beta + 1)} \pi^{2 \beta+4}.
\end{equation*}
\end{proof}

\section{Extension of the previous results to general quantum channels} \label{sec:quantum-channels}

The purpose of this section is to present an extension of the main results obtained in Sections \ref{sec:data_processing} and \ref{sec:maximal_f_divergences} to general quantum channels.  To this end, we will adopt the following strategy: First, we will see that our results extend to states which are not full rank. Then, we will use a Stinespring dilation to lift our results to arbitrary quantum channels. In this case, the theorem corresponding to  the main result of Section \ref{sec:data_processing} reads as follows:

\begin{thm}\label{thm:bound-data-processing-BS-channels}
Let $\mathcal M$ and $\mathcal{N}$ be two matrix algebras and let $\mathcal T: \mathcal M \to \mathcal N$ be a completely positive trace-preserving map with $V$ the isometry from a Stinespring dilation of $\mathcal T$ (Theorem \ref{thm:stinespring}). Let $\sigma$, $\rho$ be two quantum states on $\mathcal{M}$ such that $\rho^0 = \sigma^0$. Then
\begin{equation} \label{eq:BSdata-processing-bound-channels}
\hat S_\mathrm{BS}(\sigma \| \rho) - \hat S_\mathrm{BS}(\sigma_{\mathcal T} \| \rho_{\mathcal T}) \geq \left( \frac{\pi}{4} \right)^4 \norm{\Gamma}_\infty^{-2}  \norm{V \sigma^{1/2} V^\ast  \st^{-1/2} \Gamma_\mathcal{T}^{1/2} \st^{1/2} \otimes I - V\Gamma^{1/2}  \sigma^{1/2}V^\ast}_2^4.
\end{equation}
Here, $\sigma^{-1}$ and $\sigma_{\mathcal T}^{-1}$ are the Moore-Penrose inverses if the states are not invertible. Moreover, we have
\begin{equation}\label{eq:BSdata-processing-bound-channels-2}
\hat S_\mathrm{BS}(\sigma \| \rho) - \hat S_\mathrm{BS}(\sigma_{\mathcal T}  \| \rho_{\mathcal T}) \geq   \left( \frac{\pi}{8} \right)^4 \norm{\Gamma}_\infty^{-4}  \norm{\st^{-1}}_\infty^{-2} \norm{ \sigma \mathcal{T}^* \left( \st^{-1} \rt \right) -  \rho}_2^4.
\end{equation}
\end{thm}

\begin{proof}
Let us first justify that the quantities that appear in Equation \eqref{eq:BSdata-processing-bound-channels} are well-defined for non full-rank states for the case in which the map considered is a trace-preserving conditional expectation. Let us recall that the BS-entropy for  non full-rank states $\sigma$, $\rho$ is given by 
\begin{equation*}
\hat S_{\operatorname{BS}} (\sigma \| \rho ) = \underset{\varepsilon \searrow 0}{\operatorname{lim}} \, \hat S_{\operatorname{BS}} (\sigma + \varepsilon I  \| \rho  + \varepsilon I ). 
\end{equation*}
By virtue of Douglas' lemma \cite[Theorem 1]{Douglas1966} $\rho^0 = \sigma^0$ implies $\mathcal{T}(\rho)^0 = \mathcal{T}(\sigma)^0$ for every positive map $\mathcal{T}$.  Hence, it follows from \cite[Proposition 3.29]{Hiai2017} that the left-hand side of Equation \eqref{eq:BSdata-processing-bound-channels} is also finite for non full-rank states. Furthermore, given $a,b >0$ and $\sigma>0$, $\rho >0$, we can easily see that
\begin{equation*}
\hat{S}_{\operatorname{BS}} (a \sigma \| b \rho ) = a \hat{S}_{\operatorname{BS}} (\sigma \|  \rho ) + a \log \left( \frac{a}{b} \right). 
\end{equation*}
Given a conditional expectation $\mathcal{E} : \mathcal{M} \rightarrow \mathcal{N}$, we define 
\begin{equation*}
\sigma^\varepsilon:= \frac{\sigma + \varepsilon I}{1 + \varepsilon d}, \phantom{asda} \rho^\varepsilon:= \frac{\rho + \varepsilon I}{1 + \varepsilon d}, \phantom{asda} \sigma_\mathcal{N}^\nu := \frac{\sn + \nu I}{1 + \nu d}, \phantom{asda}  \rho_\mathcal{N}^\nu := \frac{\rn + \nu I}{1 + \nu d}.
\end{equation*}
Here, we have assumed that the identity in $\mathcal M$ has trace $d \in \mathbb N$. By the above, we have
\begin{align*}
\hat{S}_{\operatorname{BS}} (\sigma \| \rho ) - \hat{S}_{\operatorname{BS}} (\sn\| \rn ) & =  \underset{\varepsilon \searrow 0}{\operatorname{lim}} \, \hat S_{\operatorname{BS}} (\sigma + \varepsilon I  \| \rho  + \varepsilon I ) - \underset{\nu \searrow 0}{\operatorname{lim}} \, \hat S_{\operatorname{BS}} (\sn + \nu I  \| \rn  + \nu I ) \\
& =  \underset{\varepsilon \searrow 0}{\operatorname{lim}} \, \underset{\nu \searrow 0}{\operatorname{lim}} \, \left[ (1+d\varepsilon)\hat S_{\operatorname{BS}} (\sigma^\varepsilon  \| \rho^\varepsilon )- (1+d\nu) \hat S_{\operatorname{BS}} (\sn^\nu  \| \rn^\nu ) \right],
\end{align*}
where we can choose $\varepsilon = \nu$  in particular.   For $\sigma^\varepsilon, \rho^\varepsilon$, Equation \eqref{eq:BSdata-processing-bound} reads as
\begin{equation*}
\hat S_\mathrm{BS}(\sigma^\varepsilon \| \rho^\varepsilon) - \hat S_\mathrm{BS}(\sn^\varepsilon \| \rn^\varepsilon) \geq \left( \frac{\pi}{4} \right)^4 \norm{\Gamma^\varepsilon}_\infty^{-2}  \norm{ \left( \sigma^\varepsilon \right)^{1/2} \left( \sn^\varepsilon \right)^{-1/2} \left( \Gamma_\mathcal{N}^\varepsilon \right)^{1/2} \left( \sn^\varepsilon \right)^{1/2} - \left( \Gamma^\varepsilon \right)^{1/2} \left( \sigma^\varepsilon \right)^{1/2}}_2^4, 
\end{equation*}
where $\Gamma^\varepsilon := \left( \sigma^\varepsilon \right)^{-1/2} \rho^\varepsilon \left( \sigma^\varepsilon \right)^{-1/2}$ and $\Gamma_\mathcal{N}^\varepsilon := \left( \sn^\varepsilon \right)^{-1/2} \rn^\varepsilon \left( \sn^\varepsilon \right)^{-1/2}$. The only thing left to do is to write the right-hand side of the expression above in terms of $\sigma$ and $\rho$. However, expanding $\sigma^\varepsilon$ and $\rho^\varepsilon$ in the right basis, if we write $P = \sigma^0$, one can show that $\Gamma^\varepsilon$ converges to $\Gamma_P \oplus I$, where $\Gamma_P = \left( \sigma |_{P} \right)^{-1/2} \rho |_{P} \left( \sigma |_{P} \right)^{-1/2}$ and we identify $P$ with the subspace it projects onto. Moreover, we can see using the spectral decomposition  of $\sigma$ that $\norm{\Gamma_P}_\infty \geq 1$, such that
\begin{equation*}
\norm{\Gamma_P \oplus I}_\infty = \norm{\Gamma_P}_\infty
\end{equation*}
and
\begin{equation*}
\underset{\varepsilon \searrow 0}{\operatorname{lim}} \, (\Gamma^\epsilon)^{1/2} = \Gamma_P^{1/2} \oplus I.
\end{equation*}
Similar considerations lead to
\begin{equation} \label{eq:singular_states}
\hat S_\mathrm{BS}(\sigma \| \rho) - \hat S_\mathrm{BS}(\sigma_{\mathcal N} \| \rho_{\mathcal N}) \geq \left( \frac{\pi}{4} \right)^4 \norm{\Gamma}_\infty^{-2} \norm{  \sigma^{1/2}  \sn^{-1/2}  \Gamma_\mathcal{N}^{1/2}  \sn^{1/2} -  \Gamma^{1/2}  \sigma^{1/2}}_2^4,
\end{equation}
where the states $\sigma$ and $\rho$ are not necessarily full-rank anymore, and thus the inverses are now Moore-Penrose inverses. 
Now, following the steps of \cite{Wilde2018},  we are in position to apply Stinespring's dilation theorem (Theorem \ref{thm:stinespring}).

 Given $\omega$ and $\tau$ states on $\mathcal M$ such that $\omega^0 = \tau^0$, let us define 
\begin{equation*}
\sigma := V \omega V^*,
\end{equation*}
\begin{equation*}
\rho := V \tau V^*.
\end{equation*}
Then, it is clear that $\mathcal E(\sigma) = \mathcal{T}(\omega) \otimes I/s$ and $\mathcal E(\rho)= \mathcal{T}(\tau) \otimes I/s$ for $\mathcal{E}= \tr_{\mathcal V}[\cdot] \otimes I/s$ and $\dim \mathcal V = s$.  Since $\mathcal{E}$ is a conditional expectation, the Inequality \eqref{eq:singular_states} holds for it and $\sigma$ and $\rho$ defined as above, yielding:
\begin{equation*}
\hat S_\mathrm{BS}(\omega \| \tau) - \hat S_\mathrm{BS}(\mathcal{T}(\omega) \| \mathcal{T}(\tau)) \geq \left( \frac{\pi}{4} \right)^4 \norm{\Gamma}_\infty^{-2}  \norm{ V \omega^{1/2} V^*  \omega_{\mathcal{T}}^{-1/2} \Gamma_\mathcal{T}^{1/2} \omega_{\mathcal{T}}^{1/2}  \otimes I  - V  \Gamma^{1/2}  \omega^{1/2} V^*  }_2^4,
\end{equation*}
where here we define $\Gamma:= \omega^{-1/2} \tau \omega^{-1/2} $ and $\Gamma_\mathcal{T}:= \omega^{-1/2}_{\mathcal{T}} \tau_\mathcal{T} \omega^{-1/2}_{\mathcal{T}} $ for $\omega_{\mathcal{T}} := \mathcal{T}(\omega) $ and $ \tau_\mathcal{T}:= \mathcal{T}(\omega)$,  since
\begin{equation*}
\hat S_\mathrm{BS}(\sigma \| \rho)=\hat S_\mathrm{BS}(\omega \| \tau),
\end{equation*}
\begin{equation*}
\hat S_\mathrm{BS}(\mathcal E(\sigma) \| \mathcal E(\rho)) = \hat S_\mathrm{BS}(\mathcal{T}(\omega) \| \mathcal{T}(\tau)),
\end{equation*}
 for the terms in the left-hand side. Moreover,
\begin{align*}
\norm{\sigma^{-1/2}\rho \sigma^{-1/2}}_\infty & = \norm{ V \omega^{-1/2} V^* V\tau V^* V \omega^{-1/2} V^*}_\infty \\
& = \norm{\Gamma}_\infty
\end{align*}
and 
\begin{align*}
& \norm{\sigma^{1/2} \mathcal E(\sigma)^{-1/2} \left(\mathcal E(\sigma)^{-1/2} \mathcal E(\rho) \mathcal E(\sigma)^{-1/2} \right)^{1/2} \mathcal E(\sigma)^{1/2}  - \left( \sigma^{-1/2} \rho \sigma^{-1/2} \right)^{1/2} \sigma^{1/2} }_2 \\
& \phantom{asdasdasda} = \norm{V \omega^{1/2} V^* \omega_{\mathcal{T}}^{-1/2} \left( \omega_{\mathcal{T}}^{-1/2} \tau_\mathcal{T} \omega_{\mathcal{T}}^{-1/2} \right)^{1/2} \omega_{\mathcal{T}}^{1/2} \otimes I - V \left( \omega^{-1/2} \tau \omega^{-1/2} \right)^{1/2} \omega^{1/2} V^* }_2.
\end{align*}
for the terms in the right hand-side, where we have only used the fact that $V$ is an isometry.
The second assertion follows from minor adjustments to the proof of Lemma \ref{lem:BS_data_processing_bound}.
\end{proof}

Before we can continue, we need to prove that we obtain $\hat S_f(\sigma || \rho)$ for non-invertible $\sigma$, $\rho$ from a limit of states.
\begin{prop} \label{prop:different_sequence}
Let $\mathcal M \subseteq \mathcal B(\mathcal H)$ be a matrix algebra for a Hilbert space $\mathcal H$ of dimension $d$ and let $\sigma$, $\rho$ be states on $\mathcal M$ such that $\rho^0 = \sigma^0$. Then, 
\begin{equation*}
\hat S_f(\sigma || \rho) = \underset{\varepsilon \searrow 0}{\operatorname{lim}} \, \hat S_f((\sigma + \epsilon I)/(1+d \epsilon) || (\rho + \epsilon I)/(1+d \epsilon))
\end{equation*}
for every operator convex function $f:(0,\infty) \to \mathbb R$.
\end{prop}
\begin{proof}
Proposition 3.29 of \cite{Hiai2017} asserts that $\hat S_f(\sigma || \rho)$ is finite. Let 
\begin{equation*}
P_f(A,B) :=  B^{1/2} f(B^{-1/2}AB^{-1/2}) B^{1/2} 
\end{equation*}
for positive definite $A$, $B \in \mathcal B(\mathcal H)$. This is the non-commutative perspective function defined in \cite[Equation 2.7]{Hiai2017}. Corollary 3.28 of \cite{Hiai2017} shows that for states such that $\rho^0 = \sigma^0$
\begin{equation*}
\hat S_f(\sigma || \rho) = \underset{n \to \infty}{\operatorname{lim}} \, \hat S_f(\sigma + K_n || \rho + K_n),
\end{equation*}
where $K_n \in \mathcal B(\mathcal H)$ is any sequence with $K_n \to 0$ such that $K_n \geq 0$ and $\sigma + K_n$, $\rho + K_n > 0$ for every $n \in \mathbb N$. Thus, in particular we can choose $K_n = \epsilon_n I/(1 + \epsilon_n d)$ and $\epsilon_n \to 0$.
\end{proof}

Using the same ideas that appear in the proof of the previous result together with Proposition \ref{prop:different_sequence}, we can also obtain from Theorem \ref{thm:bound-data-processing-max} the following analogous result for general quantum channels.

\begin{thm}\label{thm:bound-data-processing-max-channels}
Let $\mathcal M$ and $\mathcal{N}$ be two matrix algebras and let $\mathcal T: \mathcal M \to \mathcal N$ be a completely positive trace-preserving map with $V$ the isometry from its Stinespring dilation (Theorem \ref{thm:stinespring}). Let  $\sigma$, $\rho$ be two quantum states on $\mathcal{M}$ such that $\rho^0 = \sigma^0$ and let $f:(0,\infty) \to \mathbb R$ be an operator convex function with transpose $\tilde f$. We assume that $\tilde f$ is operator monotone decreasing and such that the measure $\mu_{- \tilde f}$ that appears in Theorem \ref{thm:operator-convex-Bhatia} is absolutely continuous with respect to  Lebesgue measure. Moreover, we assume that  for every $T \geq 1$, there exist constants $\alpha \geq 0$, $C > 0$ satisfying $\mathrm{d}\mu_{ - \tilde f} (t) /\mathrm{d}t \geq (C T^{2 \alpha})^{-1}  $ for all $t \in \left[ 1/T , T \right]$ and such that Equation \eqref{eq:not_too_far_from_dpi} holds. Then, there is a constant $K_{\alpha} > 0$ such that 
\begin{equation*} \label{eq:Max-data-processing-bound-channels}
\hat S_f(  \sigma \| \rho) - \hat S_f(\st \| \rt)   \geq \frac{K_\alpha}{C} \left( 1 +  \norm{\Gamma}_\infty\right)^{-(4 \alpha + 2)}  \norm{ V \sigma^{1/2} V^\ast  \st^{-1/2} \Gamma_\mathcal{T}^{1/2} \st^{1/2} \otimes I - V\Gamma^{1/2}  \sigma^{1/2}V^\ast}_2^{4(\alpha + 1)}.
\end{equation*}
Furthermore,  there is another constant $L_\alpha > 0$ such that
\begin{equation*} 
 \hat{S}_f (\sigma\| \rho) -\hat{S}_f (\st\| \rt)   \geq \frac{L_\alpha}{C} \left( 1 +  \norm{\Gamma}_\infty \right)^{-(4 \alpha + 2)}  \norm{\Gamma}_\infty^{-(2 \alpha + 2)} \norm{\st^{-1}}_\infty^{-(2 \alpha + 2)} \norm{ \rho - \sigma \, \mathcal{T}^* \left(  \st^{-1} \rt \right)  }_2^{4(\alpha +1)}.
\end{equation*}
Here, we consider again Moore-Penrose inverses if the states are not invertible.
\end{thm}

\begin{remark}
Notice that the procedure followed to extend Theorems \ref{thm:bound-data-processing-BS} and \ref{thm:bound-data-processing-max} to Theorems \ref{thm:bound-data-processing-BS-channels} and \ref{thm:bound-data-processing-max-channels}, respectively, consists of two main ingredients: The extension of our previous results to not necessarily full-rank states followed by Stinespring's dilation theorem. Analogously to what we have done in this section, this procedure can be also applied to the setting presented in \cite{Carlen2017a} and \cite{Carlen2018} to extend the main results therein to general quantum channels. In particular, Equation \eqref{eq:vc_main_result} holds for general quantum channels.
\end{remark}

\bigskip

\noindent {\it Acknowledgments.} 
A.B.\ acknowledges support from the ISAM Graduate Center at
the Technische Universit{\"a}t M{\"u}nchen and would like to thank the group of David P{\'e}rez-Garc{\'i}a for their hospitality. A.C.\ is partially supported by a La Caixa-Severo Ochoa grant (ICMAT Severo Ochoa project SEV-2011-0087, MINECO) and acknowledges support from MINECO (grant MTM2017-88385-P) and from Comunidad de Madrid (grant QUITEMAD-CM, ref. P2018/TCS-4342). This project has received funding from the European Research Council (ERC) under the European Union’s Horizon 2020 research and innovation programme (grant agreement No 648913). Moreover, A.C.\ would like to thank the group of Michael M.\ Wolf for their hospitality. Both authors would like to thank Mil{\'a}n Mosonyi for clarifying some aspects of maximal $f$-divergences for them and Mark Wilde for bringing \cite{Wilde2018} to their attention, which inspired them to extend their results to general quantum channels. Finally, both authors would like to thank the Institut Henri Poincar{\'e} in Paris and the organizers of the trimester on ``Analysis in Quantum Information Theory", during which the idea for this project was born.

\bibliographystyle{alpha}
\bibliography{petzlit}

\vspace{1cm}
\end{document}